\documentclass[a4paper]{article}

\usepackage{ifxetex}
\ifxetex
  \usepackage{fontspec}
\else
  \usepackage[T1]{fontenc}
  \usepackage[utf8]{inputenc}
  \usepackage{lmodern}
\fi

\usepackage{amsmath}
\usepackage{amssymb}
\usepackage{amsthm}
\usepackage{xparse}
\usepackage[left=20mm,right=20mm,top=20mm,bottom=20mm]{geometry}
\usepackage{mathtools}
\usepackage{tikz}
\usepackage{pgfplots}
\usepackage{bbm}
\usepackage{enumitem}
\usepackage{authblk}
\usepackage{todonotes}
\usepackage{subcaption}
\usepackage[binary-units]{siunitx}
\usepackage[ruled,vlined,linesnumbered]{algorithm2e}
\usepackage[colorlinks=true, allcolors=blue]{hyperref}
\usepackage[all]{hypcap}
\usepackage[capitalize,noabbrev]{cleveref}

\usetikzlibrary{positioning}
\usetikzlibrary{matrix}
\usetikzlibrary{fit}
\usetikzlibrary{patterns}

\pgfplotsset{compat=1.15}

\makeatletter
\newcommand\resetstackedplots{%
\makeatletter
\pgfplots@stacked@isfirstplottrue
\makeatother
}
\makeatother

\DeclareDocumentCommand\reviewComment{mm}{}

\newtheorem{theorem}{Theorem}
\newtheorem{problem}[theorem]{Problem}
\newtheorem{corollary}[theorem]{Corollary}
\newtheorem{lemma}[theorem]{Lemma}

\DeclareDocumentCommand\transpose{m}{#1^{\intercal}}
\DeclareDocumentCommand\N{}{\mathbb{N}}
\DeclareDocumentCommand\Z{}{\mathbb{Z}}

\DeclareDocumentCommand\R{}{\mathbb{R}}
\DeclareDocumentCommand\F{}{\mathbb{F}}
\DeclareDocumentCommand\onevec{o}{\IfNoValueTF{#1}{\mathbbm{1}}{\mathbbm{1}_{#1}}}

\DeclareDocumentCommand\conv{o}{\operatorname{conv}\IfValueTF{#1}{\left(#1\right)}{}}
\DeclareDocumentCommand\supp{o}{\operatorname{supp}\IfValueTF{#1}{\left(#1\right)}{}}

\DeclareDocumentCommand\orderO{o}{\mathcal{O}\IfValueTF{#1}{\left(#1\right)}{}}
\DeclareDocumentCommand\zerovec{o}{\IfNoValueTF{#1}{\mathbb{O}}{\mathbb{O}_{#1}}}
\DeclareDocumentCommand\onevec{o}{\IfNoValueTF{#1}{\mathbbm{1}}{\mathbbm{1}_{#1}}}
\DeclareDocumentCommand\rank{o}{\operatorname{rank}\IfValueTF{#1}{\left(#1\right)}{}}
\DeclareDocumentCommand\BG{}{\operatorname{BG}}

\title{Recognizing Series-Parallel Matrices in Linear Time}
\author[1]{Matthias Walter\thanks{m.walter@utwente.nl}}
\affil[1]{Department of Applied Mathematics, University of Twente, Enschede, The Netherlands}
\date{\today}

\begin{document}

\maketitle

\begin{abstract}
  A series-parallel matrix is a binary matrix that can be obtained from an empty matrix by successively adjoining rows or columns that are copies of an existing row/column or have at most one $1$-entry.
  Equivalently, series-parallel matrices are representation matrices of graphic matroids of series-parallel graphs, which can be recognized in linear time.
  We propose an algorithm that, for an $m$-by-$n$ matrix $A$ with $k$ nonzeros, determines in expected $\orderO(m+n+k)$ time whether $A$ is series-parallel, or returns a minimal non-series-parallel submatrix of $A$.
  We complement the developed algorithm by an efficient implementation and report about computational results.
\end{abstract}


\section{Introduction}
\label{sec_intro}

We consider binary matrices $A \in \{0,1\}^{m \times n}$ and the matroids represented by those (see the books of Oxley~\cite{Oxley92} or Truemper~\cite{Truemper98} for relevant matroid concepts).
\emph{Series-parallel matroids} are those represented by \emph{series-parallel matrices}, which are defined recursively, where $m$ and $n$ are nonnegative integers:
every binary $m$-by-$n$ matrix with $m,n \leq 1$ is series-parallel.
For $m \geq 2$ or $n \geq 2$, $A$ is series-parallel if and only if it can be obtained from another series-parallel matrix $A'$ by adjoining a row vector (resp.\ column vector) that is a copy of a row (resp.\ column) vector of $A'$ or is a unit or zero vector.
The removal of such a row or column is called an \emph{SP-reduction}, more precisely a \emph{copy reduction}, \emph{unit reduction} or \emph{zero reduction}, respectively.
In other words, a matrix is series-parallel if there is a sequence of SP-reductions that yields the empty (i.e., $0$-by-$0$) matrix.
Matrices for which no SP-reduction is possible are called \emph{SP-reduced}.
The main problem of interest is the following.

\begin{problem}
  \label{prob_sp}
  Determine whether a given binary matrix is series-parallel.
\end{problem}

Series-parallel matroids were invented as generalizations of the well-known \emph{series-parallel graphs}.
The latter are those graphs that can be obtained from the graph with one node and a loop edge by iteratively duplicating an edge or subdividing an edge by a new node.
Series-parallel graphs naturally arise in electrical networks~\cite{Duffin65} and are well studied~\cite{BrandstaedtLS99}.
Moreover, the application of series-parallel extensions (the inverse operations of SP-reductions) to uniform matroids were studied by Chaourar and Oxley~\cite{ChaourarO03}.
It is well known that there is a one-to-one correspondence between the series-parallel graphs and the series-parallel matroids.
A matrix $A \in \F^{m \times n}$ \emph{represents a matroid} over some field $\F$ in the following sense:
the ground set $E$ of the matroid is the set of columns of the matrix $[\onevec \mid A]$, where $\onevec$ is the identity matrix of order $m$, and a subset $Y \subseteq E$ is independent if the corresponding columns are linearly independent over $\F$.
A binary matrix $A$ \emph{represents a graph} $G = (V,E)$ with respect to a spanning tree $T \subseteq E$ if the rows of $A$ can be indexed by the edges in $T$, the columns of $A$ correspond to the non-tree edges $E \setminus T$ and $A_{e,f} = 1$ holds if and only if the fundamental cycle $T \cup \{f\}$ contains edge $e$.
If this is the case, a \emph{copy reduction} (resp.\ \emph{unit reduction}) of a row corresponds to a contraction of a tree edge that is in series with (resp.\ parallel to) another edge.
Conversely, a \emph{copy reduction} (resp.\ \emph{unit reduction}) of a column corresponds to the deletion of a non-tree edge that is parallel to (resp.\ in series with) another edge.
Finally, a \emph{zero reduction} of a row corresponds to the contraction of a cut edge (i.e., an edge whose removal disconnects its endnodes), while a \emph{zero reduction} of a column corresponds to the deletion of a loop.
See Chapter~4.3 in Truemper's book~\cite{Truemper90} for proofs.

Not every matrix represents a graph, but one can recognize in almost-linear time whether this is the case~\cite{Fujishige80,BixbyW88}.
Here, \emph{linear} refers to the number $k$ of nonzeros of the given matrix, where assume $k \geq m,n$ throughout the paper.
Moreover, almost linear refers to $\orderO(k \cdot \alpha(k))$, where $\alpha$ denotes the \emph{inverse Ackermann function}~\cite{Tarjan75}.
The term ``almost linear'' is justified since $\alpha(k) \leq 4$ holds for all practically relevant values of $k$, namely for $k \leq 2^{65536}$.
The recognition problem for series-parallel graphs can be solved in linear time~\cite{ValdesTL82} (in the number of nodes and edges of the graph).
This immediately yields an almost-linear-time algorithm for \cref{prob_sp} by computing a graph $G$ represented by $A$ and then testing whether $G$ is series-parallel.

\paragraph{Contribution.}
First, our work improves upon the almost linear running time by removing the factor $\alpha(k)$.
Second, it yields a much simpler algorithm since the computation of the graph $G$ is quite involved.
Third, our extended algorithm finds, for a given non-series-parallel matrix, a forbidden submatrix as a certificate for not begin series-parallel.
It is unclear how to obtain such a certificate via the alternative approach sketched above, in particular because the algorithms~\cite{Fujishige80,BixbyW88} are not certifying.
Of course, one can always construct such a matrix by successively removing rows or columns and running the recognition algorithm again, but this clearly increases the running time from linear to quadratic.
Fourth, our algorithm can -- in contrast to the previous approach --  determine a maximal sequence of SP-reductions that can be applied to a given matrix (see \cref{sec_reduction}).
This is a useful preprocessing step for every recognition problem for matrix classes that are closed under adjoining zero or unit vectors or copies of existing rows/columns.
Examples of such a matrix classes are totally unimodular~\cite[Chapters~19 and~20]{Schrijver86}, balanced~\cite{Berge72}, perfect~\cite{Padberg84} and ideal~\cite{Lehman79} matrices.
In particular, the algorithm contributes to improvements for the state-of-the-art implementation of a total unimodularity test~\cite{WalterT13}.
The latter is used by researchers in combinatorial optimization and mixed-integer optimization to analyze problem structure (see, e.g.,~\cite{ChenDJ14,ConfortiFHW22,DeanGSAK22,DeukerF23,HenselmanD14}) since presence of total unimodularity indicates tractable problem (sub-)structures.
For instance, mixed-integer programming models with transportation or precedence constraints often contain such matrices.

\paragraph{Outline.}
In \cref{sec_reduction} we describe our main algorithm.
In \cref{sec_nonsp} we describe an extension that computes for a non-series-parallel matrix in linear time a minimal submatrix with the same property.
The short \cref{sec_ternary} is about the extension to ternary matrices, i.e., those with entries in $\{-1,0,+1\}$.
In \cref{sec_computations} we describe our implementation of the algorithm and report about computational results.

\section{Recognizing series-parallel matrices}
\label{sec_reduction}

\DeclareDocumentCommand\reduced{}{\mathcal{R}}
\DeclareDocumentCommand\queue{}{\mathcal{Q}}
\DeclareDocumentCommand\hashtable{}{\mathcal{H}}

The definition of series-parallel matrices is symmetric with respect to rows and columns.
Hence, we call their (disjoint) union $A$'s \emph{elements} and denote them by $E$.
For an element $e \in E$, $A(e)$ denotes the row vector $A_{r,\star}$ if $e$ is row $r$, while it denotes the column vector $A_{\star,c}$ if $e$ is column $c$.

We will present an algorithm that sequentially removes elements from the input matrix $A$ until it is SP-reduced.
It is easy to see that if an SP-reduction for element $e \in E$ is possible, and another one for $e' \in E$ is carried out, then an SP-reduction for $e$ will also be possible for the reduced matrix.
This shows that the order of removal does not matter, and hence $A$ is series-parallel if and only if the SP-reduction procedure terminates with an empty matrix.
Due to the simplicity of this algorithm, the only challenge lies in the running time.

\subsection{Data structures}
In order to achieve its running time, our algorithm relies on a couple of data structures.
First, in order to efficiently carry out a sequence of SP-reductions, we store the nonzeros of $A$ in a grid of doubly-linked lists.
More precisely, for each nonzero we store pointers to the previous and next nonzeros in the same row and to those in the same column, respectively.
We assume that the input matrix $A$ is given in a form that allows the creation of this data structure in linear time.
For instance, this is the case if the nonzeros are given as a list that is ordered lexicographically by rows and columns.
Moreover, since an SP-reduction would formally cause re-numbering of rows or columns, we actually replace nonzero entries by zeros.

Second, for each element $e \in E$, we store the number of nonzeros of $A(e)$, denoted by $|A(e)|_1$ for convenience.
This allows us to identify zero or unit reductions in constant time.

Third, we maintain a queue $\queue$ that contains all candidate elements for SP-reductions.
The main iteration of the algorithm consists in finding out whether an element extracted from $\queue$ admits an SP-reduction.
If this is the case, the reduction will be carried out, which may imply the addition of other elements to the queue.

Fourth, a hash table $\hashtable$ is used in order to identify copy reductions in constant time.
The corresponding hash function $h : E \to \Z$ shall depend on $A(e)$ only.
Moreover, we frequently update the hash value of an element $e \in E$ after a nonzero entry has been removed, which means that after one entry of $A(e)$ is modified, re-computing $h(e)$ shall be done in constant time.

We present such a hash function $h$ that requires randomization.
Let $(p,q) \in \N^n \times \N^m$ be a vector obtained by rounding a vector randomly chosen from a sphere in $\R^{m + n}$ of sufficiently large radius $R$, intersected with the first orthant.
We define
\begin{equation}
  h(e) \coloneqq \begin{cases}
                    \transpose{p} \transpose{A(e)} & \text{if $e$ is a row element,} \\
                    \transpose{q} A(e) & \text{if $e$ is a column element.}
                 \end{cases}
\end{equation}
Notice that for row elements $e \in E$, $A(e)$ is a row vector, while it is a column vector for column elements.
By the choice of $p$ and $q$, $h(e)$ is almost-surely collision-free for large radius $R$.
Moreover, if a $1$-entry of $A(e)$ is turned into a $0$-entry, $h(e)$ decreases by a corresponding entry of $p$ or $q$.
Hence, the hash value of an element can be updated in constant time.

\pagebreak[4]

\subsection{Reduction algorithm}

With these data structures at hand, we can now state our recognition algorithm.

\bigskip

\begin{algorithm}[H]
  \DontPrintSemicolon
  \SetAlgoLined
  \KwIn{Matrix $A \in \Z^{m \times n}$}
  \KwOut{Maximal sequence of SP-reductions \newline}
  Initialize list representation of $A$. \;
  Initialize empty hash table $\hashtable$ for keys $e \in E$ and values $A(e)$ and compute $h(e)$ for all $e \in E$. \;
  Initialize queue $\queue$ with all $e \in E$. \;
  Initialize the set $\reduced := \varnothing$ of recorded SP-reductions. \;
  \While{$\queue$ is not empty}
  {
    Extract element $e$ from $\queue$. \label{algo_reduce_extract} \;
    \uIf{$|A(e)|_1 = 0$}
    {
      Add $e$ to $\reduced$ and mark it as \emph{zero reduction}. \label{algo_reduce_zero_record} \;
    }
    \uElseIf{$|A(e)|_1 = 1$}
    {
      Add $e$ to $\reduced$ and mark it as \emph{unit reduction}. \label{algo_reduce_unit_record} \;
      Let $f \in E$ be such that $\{e,f\}$ are row and column indices of the $1$-entry of $A(e)$. \;
      Remove nonzero $\{e,f\}$ from $A$, add $f$ to $\queue$ if necessary, update hash value of $f$, and remove $f$ from $\hashtable$ if necessary. \label{algo_reduce_unit_nonzero} \;
    }
    \Else
    {
      Check via $\hashtable$ whether there is an element $e' \in E$ such that $A(e) = A(e')$. \label{algo_reduce_copy_check} \;
      \uIf{$\hashtable$ contains element $e' \in E$ with $A(e) = A(e')$}
      {
        Add $e$ to $\reduced$ and mark it as \emph{copy reduction for $e'$}. \label{algo_reduce_copy_record} \;
        \For{each $f$ such that $A(e)_f = 1$}
        {
          Remove nonzero $\{e,f\}$ from $A$, add $f$ to $\queue$ if necessary, update hash value of $f$, and remove $f$ from $\hashtable$ if necessary. \label{algo_reduce_copy_nonzero}
        }
      }
      \Else
      {
        Add $e$ to $\hashtable$. \label{algo_reduce_add_to_hash}
      }
    }
  }
  \Return{$\reduced$}
  \caption{Finding a maximal sequence of SP-reductions.}
  \label{algo_reduce}
\end{algorithm}

\bigskip

\begin{theorem}
  \label{thm_algo_reduce}
  For input matrices with $k$ nonzeros, \cref{algo_reduce} finds a maximal sequence of SP-reductions in expected $\orderO(k)$ time.
\end{theorem}

\begin{proof}
  We first show that the algorithm actually finds a maximal sequence of SP-reductions.
  It is easy to see that all modifications of $A$ correctly reflect the recorded SP-reductions.
  We claim that the following invariants are satisfied for all elements $e \in E$ throughout the algorithm:
  \begin{enumerate}[label=(\roman*)]
  \item
    \label{eq_algo_reduce_partition}
    either $e \in \queue$ or $e \in \hashtable$ or $e \in \reduced$;
  \item
    \label{eq_algo_reduce_reducible}
    if an SP-reduction for $e$ is possible for $A$, then $e' \in \queue$ holds for
    some element with $A(e') = A(e)$.
  \end{enumerate}
  Since we initialize $\queue$ as $E$, both statements are satisfied at the beginning.
  Now consider an iteration of the main loop in which element $e \in E$ was extracted from $\queue$.
  Either $e$ is added to $\reduced$ in line~\ref{algo_reduce_zero_record}, line~\ref{algo_reduce_unit_record} or line~\ref{algo_reduce_copy_record} or $e$ is added to $\hashtable$ in line~\ref{algo_reduce_add_to_hash}.
  Moreover, if the SP-reduction causes the removal of nonzeros $\{e,f\}$ in line~\ref{algo_reduce_unit_nonzero} or line~\ref{algo_reduce_copy_nonzero}, then it is ensured that $f \in \queue$ and $f \notin \hashtable$ hold.
  This establishes invariant~\ref{eq_algo_reduce_partition}.

  Consider, for the sake of contradiction, a first iteration after which invariant~\ref{eq_algo_reduce_reducible} is violated, i.e.,
  an SP-reduction for $\hat{e}$ is possible, but no $e' \in E$ with $A(e') = A(\hat{e})$ is in the queue.
  Moreover, for each such $e'$ we have $e' \notin \reduced$ and thus $e' \in \hashtable$ by the invariant~\ref{eq_algo_reduce_partition}.
  If $\hat{e}$ was extracted from $\queue$ in this iteration, i.e., $\hat{e} = e$ holds, then we must have added $\hat{e}$ to $\hashtable$ in line~\ref{algo_reduce_add_to_hash}.
  In particular, the SP-reduction must be a copy reduction for some other element $e' \in E$.
  Since we argued that $e' \in \hashtable$ holds, we obtain a contradiction to the fact that we reached line~\ref{algo_reduce_add_to_hash}.
  Otherwise, $\hat{e}$ must have become SP-reducible, i.e., $\hat{e} = f$ holds for some element $f \in E$ for which $\{e,f\}$ is a nonzero of $A$.
  However, in the corresponding lines~\ref{algo_reduce_unit_nonzero} and~\ref{algo_reduce_copy_nonzero}, such elements $f$ are added to $\queue$, which yields a contradiction.
  We conclude that also invariant~\ref{eq_algo_reduce_reducible} holds.

  The total number of iterations is bounded by $m + n + k$ since $|\queue| = |E| = m + n$ holds initially, and since further additions to $\queue$ happen at most once per (removed) nonzero.
  This shows that the algorithm terminates, and by invariants~\ref{eq_algo_reduce_partition} and~\ref{eq_algo_reduce_reducible} with maximal $\reduced$.
  It also shows that lines~\ref{algo_reduce_extract}, \ref{algo_reduce_zero_record}, \ref{algo_reduce_unit_record}, \ref{algo_reduce_unit_nonzero}, \ref{algo_reduce_copy_check}, \ref{algo_reduce_copy_record} and~\ref{algo_reduce_add_to_hash} are each executed at most $m+n+k$ times.
  Clearly, lines~\ref{algo_reduce_copy_record} and~\ref{algo_reduce_copy_nonzero} are executed at most $k$ times since each time a nonzero is removed from $A$.
  Due to the data structures and the properties of the hash function, each of these lines can be executed in constant time, where this holds for lines~\ref{algo_reduce_unit_nonzero} and~\ref{algo_reduce_copy_nonzero} almost surely.
  We conclude that the overall running time is linear in $m + n + k$ in expectation.
\end{proof}

\section{Certifying non-series-parallel matrices}
\label{sec_nonsp}

In case a given matrix $A$ is series-parallel, the list of SP-reductions produced by \cref{algo_reduce} yields a certificate, i.e., an easily verifiable reason for begin series-parallel.
However, in the case that $A$ is not series-parallel, a user has to trust the correctness of its implementation.
Hence, it is desirable to be able to provide a simple reason for this negative outcome as well.
The goal of this section is to extend the algorithm as to provide such a certificate in terms of forbidden submatrices.
The latter are minimal non-series-parallel matrices, i.e., matrices that are not series-parallel, but for which every proper submatrix is series-parallel.

These matrices also appear in the context of the recognition of graphic matroids.
A \emph{graphic matroid} is defined by a connected undirected graph $G = (V, E)$ in which a subset $F \subseteq E$ of edges is \emph{independent} if it does not contain a cycle.
The \emph{bases} of the matroid are the maximal independent sets, that is, the spanning trees of $G$.
The graph $G$ together with one of its spanning trees $T$ defines a \emph{representation matrix} $M(G,T) \in \{0,1\}^{T \times (E \setminus T)}$ via $M(G,T)_{e,f} = 1$ holds if and only if the tree edge $e \in T$ lies on the unique cycle in $T \cup \{f\}$.

It is worth to mention that in this context the SP-reductions correspond directly to the graph operations.
For instance, the addition of a unit column (indexed by $f$) with the $1$-entry in row $e$ corresponds to the addition of an edge $f$ that is parallel to the tree edge $e$.
Similarly, the addition of a copy $f$ of a row $e$ corresponds to a replacement of $e$ by the a path of length $2$, i.e., $e$ and $f$ will be in series.

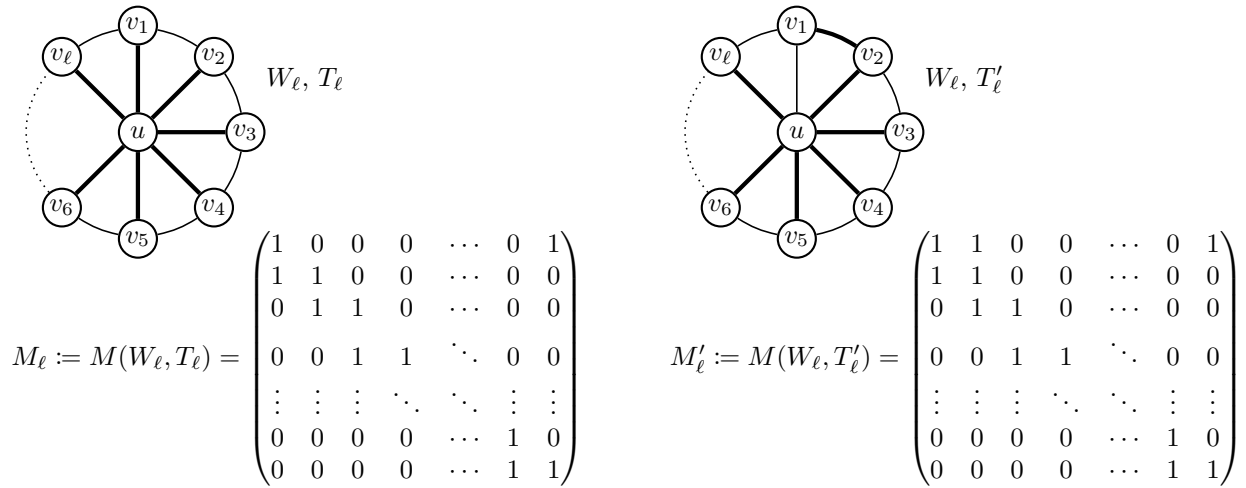
\begin{figure}[htpb]
  \hfill
  \subfloat[Wheel graph $W_\ell$ with spanning tree $T_\ell$ that consists of all spoke edges, together with the representation matrix.\label{fig_wheel_graphs1}]{%
    \begin{tikzpicture}[
      node/.style={circle,draw=black,inner sep=0mm,minimum size=5mm,thick},
      ]

      \foreach \v/\x/\y/\l in {u/0/0/$u$, v1/0/1.414/$v_1$, v2/1/1/$v_2$, v3/1.414/0/$v_3$,%
          v4/1/-1/$v_4$, v5/0/-1.414/$v_5$, v6/-1/-1/$v_6$, vl/-1/1/$v_\ell$}{
        \node[node] (\v) at (\x,\y) {\l};
      }
      \foreach \v in {v1, v2, v3, v4, v5, v6, vl}{
        \draw[ultra thick] (u) to (\v);
      }
      \foreach \v/\w in {vl/v1, v1/v2, v2/v3, v3/v4, v4/v5, v5/v6}{
        \draw[semithick] (\v) to[bend left=10] (\w);
      }
      \draw[semithick,dotted] (v6) to[bend left=40] (vl);
      \node at (2.2,0.7) {$W_\ell$, $T_\ell$};

      \node at (2.1,-3) {$
        M_\ell \coloneqq M(W_\ell,T_\ell) = \begin{pmatrix}
          1 & 0 & 0 & 0 & \dotsb & 0 & 1 \\
          1 & 1 & 0 & 0 & \dotsb & 0 & 0 \\
          0 & 1 & 1 & 0 & \dotsb & 0 & 0 \\
          0 & 0 & 1 & 1 & \ddots & 0 & 0 \\
          \vdots & \vdots & \vdots & \ddots & \ddots & \vdots & \vdots \\
          0 & 0 & 0 & 0 & \dotsb & 1 & 0 \\
          0 & 0 & 0 & 0 & \dotsb & 1 & 1
        \end{pmatrix}$};
    \end{tikzpicture}
  }
  \hfill
  \hspace{4mm}
  \hfill
  \subfloat[Wheel graph $W_\ell$ with spanning tree $T'_\ell$ that consists of all but one spoke edge as well as one cycle edge, together with the representation matrix.\label{fig_wheel_graphs2}]{%
    \begin{tikzpicture}[
      node/.style={circle,draw=black,inner sep=0mm,minimum size=5mm,thick},
      ]

      \foreach \v/\x/\y/\l in {u/0/0/$u$, v1/0/1.414/$v_1$, v2/1/1/$v_2$, v3/1.414/0/$v_3$,%
          v4/1/-1/$v_4$, v5/0/-1.414/$v_5$, v6/-1/-1/$v_6$, vl/-1/1/$v_\ell$}{
        \node[node] (\v) at (\x,\y) {\l};
      }
      \foreach \v in {v2, v3, v4, v5, v6, vl}{
        \draw[ultra thick] (u) to (\v);
      }
      \draw[semithick] (u) to (v1);
      \draw[ultra thick] (v1) to[bend left=10] (v2);
      \foreach \v/\w in {vl/v1, v1/v2, v2/v3, v3/v4, v4/v5, v5/v6}{
        \draw[semithick] (\v) to[bend left=10] (\w);
      }
      \draw[semithick,dotted] (v6) to[bend left=40] (vl);
      \node at (2.2,0.7) {$W_\ell$, $T_\ell'$};
      \node at (2.1,-3) {$
        M_\ell' \coloneqq M(W_\ell,T_\ell') = \begin{pmatrix}
          1 & 1 & 0 & 0 & \dotsb & 0 & 1 \\
          1 & 1 & 0 & 0 & \dotsb & 0 & 0 \\
          0 & 1 & 1 & 0 & \dotsb & 0 & 0 \\
          0 & 0 & 1 & 1 & \ddots & 0 & 0 \\
          \vdots & \vdots & \vdots & \ddots & \ddots & \vdots & \vdots \\
          0 & 0 & 0 & 0 & \dotsb & 1 & 0 \\
          0 & 0 & 0 & 0 & \dotsb & 1 & 1
        \end{pmatrix}$};
    \end{tikzpicture}
  }
  \hfill
  \phantom{x}
  \caption{Wheel graph $W_\ell$ with two different spanning trees $T_\ell$ and $T_\ell'$ as well as the two corresponding representation matrices $M_\ell$ and $M_\ell'$.}
  \label{fig_wheel_graphs}
\end{figure}

In \cref{fig_wheel_graphs} so-called \emph{wheel graphs} and the \emph{wheel matrices} of order $\ell$ are depicted.
They consist of a cycle of length $\ell \geq 3$ plus one node that is connected to every other node via \emph{spoke} edges.
In the figure, two different spanning trees $T_\ell$ and $T_\ell'$ together with the representation matrices $M_\ell$ and $M_\ell'$ are depicted.

These matrices are required as a starting point in Truemper's total unimodularity testing algorithm~\cite{Truemper90}.
There, in order to test for graphicness, a submatrix representing $W_3$ must be found.
The intuitive reason is that $W_3$ is the smallest 3-connected graph, and Truemper's graphicness test first finds such a submatrix and then iteratively grows it to a sequence of larger (graphic) submatrices.
This sequence is later used to efficiently search for certain matrix decompositions.
This \emph{graph minor} $W_3$ is determined from any $W_\ell$ by successively deleting a spoke edge and contracting a cycle edge.

Notice that the matrices $M_\ell$ and $M_\ell'$ differ only in the entry in the first row and second column.
It is easy to see that both matrices are SP-reduced.
We say that a matrix \emph{contains a wheel submatrix} if it contains a wheel matrix as a submatrix, possibly after permuting rows or columns.

We will show that the following holds.
\begin{theorem}
  \label{thm_characterization}
  A binary matrix is either series-parallel or it contains a wheel submatrix.
\end{theorem}
The proof is delayed as it follows from the correctness of our algorithm that returns one of these matrices as a submatrix when confronted with a matrix that is not series-parallel (see \cref{thm_algo_wheel_search}).
Using matroid operations, the certificate can be simplified even further.
The following corollary is a strengthening of Theorem~4.2.13 in Truemper's book~\cite{Truemper98}, where only the case of connected matroids is discussed.

\begin{corollary}
  A binary matroid is either series-parallel or contains the graphic matroid of the wheel graph $W_3$ as a minor.
\end{corollary}

\begin{proof}
  The result follows from \cref{thm_characterization} by observing that a graphic matroid of $W_\ell$ with $\ell \geq 3$ contains the graphic matroid of $W_3$ as a minor.
  In fact, a binary pivot operation on the entry in the second row and first column of $M_\ell$ yields $M_\ell'$, and $M_\ell'$ contains $M_{\ell-1}$ as a submatrix.
  Repeated application yields $M_3$ after $\ell-3$ pivots.
\end{proof}

\subsection{Bipartite graph}
\label{sec_nonsp_bipartite_graph}

We now introduce a graph-theoretic viewpoint that is important for the detection of wheel matrices.
A binary matrix $A \in \{0,1\}^{m \times n}$ gives rise to a \emph{bipartite graph}, denoted by $\BG(A)$, which has $m$ nodes on one side $R$ and $n$ nodes on the other side $C$ of the bipartition and those edges $\{r,c\}$ (with $r \in R$, $c \in C$) for which $A_{r,c} = 1$.
It is easy to see that $\BG(M_\ell)$ is a chordless cycle of length $2\ell$.
Hence, the basic idea of our recognition algorithm is to first apply \cref{algo_reduce} and then to find a chordless cycle in $\BG(A)$ of length at least $6$ in the SP-reduced matrix $A$.
Chordless cycles can be found using breadth-first search.
However, it may turn out that all found cycles have length $4$, which corresponds to a $2$-by-$2$ matrix with only $1$s.

In his paper~\cite{Truemper90}, Truemper describes a way to enforce finding a longer cycle (if one exists):
first, one grows the submatrix consisting of $1$s to an inclusion-wise maximal one, indexed by rows $X$ and columns $Y$.
Then, one searches for a shortest path $P$ from $X$ to $Y$ in $\BG(A)$ without using an edge corresponding to a $1$ in this submatrix.
If $P$ exists, then the submatrix induced by $P$ and one additional row and column is of type $M_{\ell}'$ (the top-left $2$-by-$2$ submatrix of $M_{\ell}'$ is part of the all-$1$s submatrix).

\begin{figure}[htb]
  \begin{equation*}
    \begin{pmatrix}
      1 & 1 & 1 & 1   & 0   & \textcolor{red}1 & 0 & 0 & \textcolor{red}1 \\
      1 & 1 & 1 & 1   & 0   & \textcolor{red}1 & 0 & 0 & 0 \\
      1 & 1 & 1 & 1   & 0   & 0 & 0 & 0 & 0 \\
      0 & 0 & 0 & 0   & 0   & \textcolor{red}1 & \textcolor{red}1 & 0 & 0 \\
      0 & 0 & 0 & 0   & 0   & 0 & \textcolor{red}1 & \textcolor{red}1 & 0 \\
      0 & 0 & 0 & 0   & 0   & 0 & 0 & \textcolor{red}1 & \textcolor{red}1 \\
      0 & 0 & 0 & 0   & 0   & 0 & \textcolor{red}1 & 0 & \textcolor{red}1 \\
      \textcolor{blue}1 & \textcolor{blue}1 & 0 & 0   & 0   & 0 & 0 & 0 & 0 \\
      0 & \textcolor{blue}1 & \textcolor{blue}1 & 0   & \textcolor{blue}1   & 0 & 0 & 0 & 0 \\
      \textcolor{blue}1 & \textcolor{blue}1 & 0 & 0   & \textcolor{blue}1   & 0 & 0 & 0 & 0
    \end{pmatrix}
  \end{equation*}
  \caption{An SP-reduced matrix $A$ with inclusion-wise maximal all-$1$s submatrix in the upper left part.
  Every path in $\BG(A)$ from the first three rows $X$ to the first four columns $Y$ must use an edge from this submatrix.
  The edges corresponding to the part reachable from $X$ are colored red, while those reachable from $Y$ are colored blue.}
  \label{fig_bipartite_graph_no_path}
\end{figure}

\pagebreak[4]

\subsection{Separations}
\label{sec_nonsp_separations}

In case such a path does not exist, we have found a \emph{$2$-separation} of $A$, which is a partitioning of $A$'s rows into $X^1$ and $X^2$ and $A$'s columns into $Y^1$ and $Y^2$ such that $\rank(A_{X^1,Y^2}) + \rank(A_{X^2,Y^1}) = 1$ and $|X^i|+|Y^i| \geq 2$ holds for $i=1,2$.
After reordering of rows and columns, $A$ looks as in \cref{fig_two_separation}.
The decomposition of $A$ into $A^1$ and $A^2$ is called a \emph{$2$-sum decomposition}.
For matroids, it corresponds to a $2$-sum decomposition involving the corresponding represented matroids.
However, we will work only on the matrix level and never exploit any matroid structure.
The $2$-separation for the example in \cref{fig_bipartite_graph_no_path} can be readily seen:
the rank-1 submatrix is the $7$-by-$5$ submatrix in the upper left, $B$ is the (smallest) submatrix containing all red $1$s, while $C$ is the (smallest) submatrix containing all blue $1$s.

\begin{figure}[htb]
  \begin{center}
    \begin{tikzpicture}
      \matrix (mat) [matrix of math nodes, inner sep=0, column sep=0, nodes={inner sep=3mm,text height=1em, text width=2em,align=center},ampersand replacement=\&,nodes in empty cells]
      {%
        B \& b \transpose{c} \\
        \zerovec \& C \\
      };
      \draw[very thick] (mat-1-1.north west) rectangle (mat-2-2.south east);
      \draw (mat-1-2.north west) -- (mat-2-2.south west);
      \draw (mat-1-1.south west) -- (mat-1-2.south east);
      \node[left=2mm of mat-1-1.south west] {$A = $};

      \matrix (matB) [right=20mm of mat, matrix of math nodes, inner sep=0, column sep=0, nodes={inner sep=3mm,text height=1em, text width=1em,align=center},ampersand replacement=\&,nodes in empty cells]
      {%
        B \& |[inner xsep=0mm]|b \\
      };
      \draw[very thick] (matB-1-1.north west) rectangle (matB-1-2.south east);
      \draw (matB-1-1.north east) -- (matB-1-1.south east);
      \node[left=2mm of matB-1-1.west] {$A^1 = $};

      \matrix (matC) [right=20mm of matB, matrix of math nodes, inner sep=0, column sep=0, nodes={inner sep=3mm,text height=1em, text width=1em,align=center},ampersand replacement=\&,nodes in empty cells]
      {%
        |[inner ysep=0.5mm]| \transpose{c} \\
        C \\
      };
      \draw[very thick] (matC-1-1.north west) rectangle (matC-2-1.south east);
      \draw (matC-1-1.south west) -- (matC-1-1.south east);
      \node[right=8mm of matB-1-2.east] {$A^2 = $};
    \end{tikzpicture}
  \end{center}
  \caption{Partitioned version of a $2$-separable matrix $A$ that is decomposed into a $2$-sum of $A^1$ and $A^2$ at the last column and first row, respectively.
  Note that it is required that $B$ and $C$ each have at least $2$ elements and that $b$ and $c$ are nonzero, i.e., $b \transpose{c}$ has rank~1.}
  \label{fig_two_separation}
\end{figure}

One may be tempted to recursively search for a wheel submatrix in both parts of the $2$-separation.
However, this may lead to an increased running time, and it turns out that both parts will contain such a submatrix.

\begin{lemma}
  \label{thm_two_separation_part}
  For a $2$-sum decomposition of an SP-reduced matrix $A \in \{0,1\}^{m \times n}$ as in \cref{fig_two_separation}, neither of submatrices $A^1$ and $A^2$ is series-parallel.
\end{lemma}

\begin{proof}
  By symmetry it suffices to prove the statement for $A^1$.
  For the sake of contradiction, assume that $A^1 = [B \mid b]$ is series-parallel and that among all $2$-separations of matrices $A$ with the same number of elements, $A^1$ has a minimum number of elements.
  Note that $A^1$ has at least $3$ rows and $3$ columns since otherwise there would be an SP-reduction applicable to $A$.
  Because $A$ is SP-reduced, the only SP-reductions applicable to $[B \mid b]$ can be column reductions that involve $b$ or a unit row reduction with a $1$-entry in $b$.
  We distinguish the possible reductions.

\medskip
\noindent
  \textbf{Case 1: $b$ is identical to a column $d$ of $B$.}
  We can remove the column $d$ from $B$ and attach it to $C$.
  Then the lower-left submatrix (in \cref{fig_two_separation}) remains a zero matrix and the upper-right matrix remains a rank-1 submatrix, just with one more column than in the given $2$-separation.
  This yields another $2$-separation, which violates our assumption on the size of $A^1$.

\medskip
\noindent
  \textbf{Case 2: $b$ is a unit vector.}
  We can remove the row $r$ in which $b$ has the unique $1$-entry from $B$ and attach it to $C$.
  This turns the upper-right submatrix into a zero matrix and the lower-left one into a rank-1 submatrix with nonzeros
  only the row $r$.
  This yields a $2$-separation of $\bar{A} = \transpose{A}$ into $\bar{A}^1$ and $\bar{A}^2$ such that $\bar{A}^1$ has one element less than $A^1$, which contradicts our assumption on $A$ and the $2$-separation.

\medskip
\noindent
  \textbf{Case 3: for some row $r'$, $b_{r'} = 1$ and $B_{r',\star} = \zerovec$ hold.}
  We can remove the row $r'$, which effectively sets $b_{r'} = 0$.
  Unless Case~1 or Case~2 were already applicable before, they must be applicable now, since otherwise there is no SP-reduction possible.
  In both cases, we also attach row $r'$ to $C$ (in addition to column $c$ or row $r$, respectively).

\medskip
\noindent
  We conclude that such a matrix $A$ with such a $2$-separation cannot exist, which completes the proof.
\end{proof}

Hence, it suffices to only consider the smaller of the two components $A^1$, $A^2$ for a recursive search.
Note that $A^1$ (or $A^2$) is not necessarily SP-reduced, and hence we need to apply \cref{algo_reduce} again.

\subsection{Wheel search algorithm}
\label{sec_nonsp_algorithm}

The previous discussion leads to the following recursive algorithm for searching a wheel submatrix.

\bigskip

\begin{algorithm}[H]
  \DontPrintSemicolon
  \SetAlgoLined
  \KwIn{Matrix $A \in \{0,1\}^{m \times n}$}
  \KwOut{Either ``$A$ is series-parallel'' together with a list of $m+n$ SP-reductions, \newline or ``$A$ is not series-parallel'' together with a wheel submatrix of $A$. \newline}
  Run \cref{algo_reduce} for $A$, obtain $\reduced$ and replace $A$ by the reduced matrix. \label{algo_wheel_search_reduce} \;
  \lIf{$|\reduced|=m+n$}{\Return{``$A$ is series-parallel'' \textnormal{together with} $\reduced$.}}
  \Else
  {
    Let $A'$ be the SP-reduced matrix. \;
    Run breadth-first search in $\BG(A')$ to find a chordless cycle $C$ of length $2 \ell$ for some $\ell \in \N$. \label{algo_wheel_search_first_bfs} \;
    \eIf{$C$ exists}
    {
      Let $B$ be the submatrix of $A$ indexed by all rows and columns of $C$. \;
      \lIf{$\ell \geq 3$}
      {
        \Return{``$A$ is not series-parallel'' \textnormal{together with} $B$.} \label{algo_wheel_search_long_cycle}
      }
      \Else
      {
        Grow $B$ to a maximal submatrix of $A'$ that contains only $1$s and let $X$ and $Y$ be the row and column sets of $B$, respectively. \;
        Define $A''$ to be $A'$ with the entries of $B$ replaced by $0$s. \;
        Using breadth-first search, search for a (shortest) path $P$ from $X$ to $Y$ in $\BG(A'')$. \label{algo_wheel_search_second_bfs} \;
        \uIf{$P$ exists}
        {
          Let $c$ be the column that comes directly after $P$'s starting node from $X$. \label{algo_wheel_search_path_column} \;
          Let $r' \in X$ be such that $A_{r',c} = 0$. \;
          Let $r$ be the row that comes directly before $P$'s end node from $Y$. \label{algo_wheel_search_path_row} \;
          Let $c' \in Y$ be such that $A_{r,c'} = 0$. \label{algo_wheel_search_path_row_zero} \;
          Let $B'$ be the submatrix indexed by all rows and columns of $P$, row $r'$ and column $c'$. \;
          \Return{``$A$ is not series-parallel'' \textnormal{together with} $B'$.}
        }
        \Else
        {
          The nodes reachable from $X$ induce a $2$-separation of $A$ with parts $A^1$ and $A^2$. \;
          Let $i \in \{1,2\}$ be such that $A^i$ has the minimum number of elements. \;
          \Return{} output of recursive call of \cref{algo_wheel_search} for $A^i$. \label{algo_wheel_search_recurse1}
        }
      }
    }
    {
      The nodes reachable from the source node of the search induce a $2$-separation of $A$ with parts $A^1$ and $A^2$. \;
      Let $i \in \{1,2\}$ be such that $A^i$ has the minimum number of elements. \;
      \Return{} output of recursive call of \cref{algo_wheel_search} for $A^i$. \label{algo_wheel_search_recurse2}
    }
  }
  \caption{Certifying recognition algorithm for binary series-parallel matrices.}
  \label{algo_wheel_search}
\end{algorithm}

\bigskip

\begin{theorem}
  \label{thm_algo_wheel_search}
  Let $A \in \{0,1\}^{m \times n}$ have $k$ nonzeros.
  Then \cref{algo_wheel_search} determines in expected $\orderO(k)$ time whether $A$ is series-parallel and if not, finds a wheel submatrix of $A$.
\end{theorem}

\begin{proof}
  We first discuss the correctness of the algorithm and then turn to its running time.

  After line~\ref{algo_wheel_search_reduce}, $A$ is SP-reduced, which implies that $\BG(A)$ is not a forest.
  Hence, the breadth-first search in line~\ref{algo_wheel_search_first_bfs} finds a chordless cycle of length at least $4$.
  As explained in \cref{sec_nonsp_bipartite_graph}, a chordless cycle of length at least $6$ in $\BG(A)$ corresponds to a submatrix $M_{\ell}$ for $\ell \geq 3$, which is returned in line~\ref{algo_wheel_search_long_cycle} if such a cycle is found.
  Otherwise, path $P$ is searched for in line~\ref{algo_wheel_search_second_bfs}.

  Suppose, $P$ exists.
  By maximality of the submatrix $B$, in line~\ref{algo_wheel_search_path_column}, $A_{X,c}$ is not the all-$1$s vector, and hence, $r'$ is well defined.
  Similarly, $c'$ from line~\ref{algo_wheel_search_path_row_zero} is well defined.
  We claim that the matrix $B'$ is of the form $M_{\ell}'$.
  To see this, observe that among the rows $X$, exactly two belong to $B'$, one with $A_{r',c} = 0$ and one with a $1$-entry in column $c$.
  Similarly, exactly two of the columns $Y$ belong to $B'$, one with $A_{r,c'} = 0$ and one with a $1$-entry in row $r$.
  The fact that $P$ is a shortest $X$-$Y$-path ensures that $B'$ has all other required $0/1$-entries.

  If $P$ does not exist, the $2$-separation of $A$ is easily verified.
  \cref{thm_two_separation_part} ensures that we can restrict our search to any of the parts $A^1$, $A^2$.
  Clearly, the recursion will terminate since $A^i$ in line~\ref{algo_wheel_search_recurse1} or line~\ref{algo_wheel_search_recurse2} has fewer elements than $A$.

  We now prove that the algorithm runs in expected linear time.
  The bipartite graph $\BG(A)$ has $m + n \in \orderO(k)$ nodes and $k$ edges, and hence the breadth-first search runs in $\orderO(k)$ time.
  Hence, all lines except for the recursive calls in lines~\ref{algo_wheel_search_recurse1} and~\ref{algo_wheel_search_recurse2} can be carried out in time $\orderO(k)$.
  To bound the overall running time, let $f(k)$ denote the running time of the algorithm for matrices with $k$ nonzeros
  (and at most $k$ rows and columns).
  Let $k_1$ and $k_2$ be the number of nonzeros of $A^1$ and $A^2$, respectively.
  We obtain
  \begin{equation*}
    k = |X| \cdot |Y| ~+~ (k_1 - |X|) ~+~ (k_2 - |Y|).
  \end{equation*}
  For the smaller of the two, $k_i$, we obtain
  \begin{equation*}
    2k_i \leq k_1 + k_2 = k - |X| \cdot |Y| + |X| + |Y| = k - (|X|-1) \cdot (|Y|-1) + 1.
  \end{equation*}
  which implies $k_i \leq k/2$ since $|X|,|Y| \geq 2$ holds.
  This yields the recurrence relation $f(k) = \orderO(k) + f(k/2)$, which yields $f(k) = \orderO(k)$.
\end{proof}

A few remarks can be made.
First, \cref{algo_wheel_search} can be modified\footnote{By removing the second row and the first column of a matrix $M_{\ell'}$.} to return an $M_{\ell-1}$ submatrix of $M_{\ell}'$ in case $\ell \geq 4$.
The modified algorithm will then either return the submatrix $M_3'$ or a submatrix $M_{\ell}$ for $\ell \geq 3$, all of which are \emph{inclusion-wise minimal} non-series-parallel submatrices.
Second, it can be modified to either return the sequence of SP-reductions in case $A$ is determined to be series-parallel, or return a wheel submatrix otherwise, making it a certifying algorithm for the recognition problem.

\begin{proof}[Proof of \cref{thm_characterization}.]
  It follows from \cref{thm_algo_wheel_search} that every matrix that is not series-parallel must contain a wheel matrix.
  Moreover, wheel matrices are SP-reduced, and hence they are not series-parallel.
  This concludes the proof.
\end{proof}

\section{Extension to ternary matrices}
\label{sec_ternary}

Scaling rows or columns of representation matrices by $-1$ does not change the underlying matroid.
While this has no effect for the binary field due to $-1 \equiv +1 \pmod 2$, it does matter for the ternary field.
Our techniques easily be extend to this case by allowing to scale rows and columns by $-1$.
Hence, a ternary series-parallel matrix is constructed from a ternary $1$-by-$1$ matrix by successively adding zero rows/columns, (negated) unit row/column vectors or (negated) copies of existing rows/columns.
We call the corresponding reductions \emph{ternary SP-reductions}.

\paragraph{Reduction algorithm.}
\cref{algo_reduce} naturally extends as well: we can replace our hash function $h: E \to \Z$ by $h' : E \to \N$ defined via $h'(e) \coloneqq |h(e)|$.
If there exist elements $e, e' \in E$ with $A(e) = -A(e')$, we will have $h(e) = -h(e')$ and thus $h'(e) = h'(e')$.
Hence, we will almost surely have collisions (only) for copies and for negated copies.
Clearly, also the check for equality of $A(e)$ and $A(e')$ has to be adapted to also detect negated copies.
However, neither of these adaptions affects the asymptotic runtime.
For the remainder of this section we thus consider \cref{algo_reduce} to be modified accordingly.

\paragraph{Certificates.}
By definition, every SP-reduction for a ternary matrix $A$ induces a corresponding SP-reduction for the (binary) support matrix $\supp(A)$.
Hence, if even $\supp(A)$ admits no (binary) SP-reduction, then any corresponding certificate also shows that $A$ is not ternary series-parallel.
In other words, for $\ell \geq 3$, the submatrices $M_\ell$ or $M_{\ell}'$ of $\supp(A)$ constitute certificates that can be found in expected linear time.

In addition, it is possible that the input matrix $A \in \{-1,0,+1\}^{m \times n}$ is ternary SP-reduced but $\supp(A)$ is not binary SP-reduced.
In this case, no zero- or unit reduction can be possible for $\supp(A)$ since the same reduction would be applicable to $A$ as well.
Hence, there must exist a binary copy reduction of $\supp(A)$ that does not correspond to a ternary copy reduction of $A$.
Without loss of generality, we consider only row copy reductions.
Thus, there exist rows $r$ and $r'$ such that $A_{r,\star} \neq A_{r',\star}$ and $A_{r,\star} \neq -A_{r',\star}$ but $\supp(A)_{r,\star} = \supp(A)_{r',\star}$ holds.
Hence, there must exist columns $c$ and $c'$ such that $A_{r,c} = A_{r',c}$ and $A_{r,c'} = -A_{r',c'}$.
This implies that $A$ contains (up to scaling of rows/columns) the matrix
\begin{align*}
  N_2 &\coloneqq \begin{pmatrix}
    -1 & 1 \\
    1 & 1
  \end{pmatrix}
\end{align*}
as a submatrix.
Clearly, $N_2$ is not series-parallel.
We call a submatrix of $A$ a \emph{signed} $B$-submatrix if it contains a submatrix that can be turned to $B$ by multiplying rows or columns with $-1$.

\bigskip

\begin{algorithm}[H]
  \DontPrintSemicolon
  \SetAlgoLined
  \KwIn{Matrix $A \in \{-1,0,1\}^{m \times n}$}
  \KwOut{Either ``$A$ is series-parallel'' together with a list of $m+n$ ternary SP-reductions, \newline or ``$A$ is not series-parallel'' together with a signed wheel- or $N_2$-submatrix of $A$. \newline}
  Run modified \cref{algo_reduce} for $A$ obtaining $\reduced$. \label{algo_ternary_search_reduce} \;
  \eIf{$|\reduced| = m + n$}
  {
    \Return{``$A$ is series-parallel'' \textnormal{together with} $\reduced$.}
  }
  {
    Let $A'$ arise from $A$ by applying all SP-reductions $\reduced$. \;
    Run \cref{algo_wheel_search} for $\supp(A')$ obtaining either $\reduced'$ or submatrix $B'$ of $\supp(A')$. \label{algo_ternary_wheel} \;
    \eIf{$\supp(A')$ is not series-parallel}
    {
      Let $B''$ be the signed wheel submatrix of $A$ corresponding to $B'$. \;
      \Return{``$A$ is not series-parallel'' \textnormal{together with} $B''$.}
    }
    {
      Let $e,e' \in E$ with $\supp(A'(e)) = \supp(A'(e'))$ be the elements of the first SP-reduction in $\reduced'$. \;
      Let $f,f' \in E$ be such that $A'(e)_f = A'(e')_f \in \{-1,+1\}$ and $A'(e)_{f'} = -A'(e')_{f'} \in \{-1,+1\}$. \;
      Let $B''$ be the submatrix of $A'$ induced by $e$, $e'$, $f$ and $f'$. \;
      \Return{``$A$ is not series-parallel'' \textnormal{together with} $B''$.}
    }
  }
  \caption{Certifying recognition algorithm for ternary series-parallel matrices.}
  \label{algo_ternary}
\end{algorithm}

\bigskip

By the discussion above, correctness of \cref{algo_ternary} follows naturally.
\begin{theorem}
  \label{thm_algo_ternary}
  Let $A \in \{-1,0,1\}^{m \times n}$ have $k$ nonzeros.
  Then \cref{algo_ternary} determines in expected $\orderO(k)$ time whether $A$ is series-parallel and if not, finds a signed wheel- or $N_2$-submatrix of $A$.
\end{theorem}

\paragraph{Structure.}
We obtain the following characterization of ternary series-parallel matrices.

\begin{theorem}
  A ternary matrix is either series-parallel or it contains a signed wheel- or $N_2$-submatrix.
\end{theorem}

The matrix $N_2$ represents the uniform matroid of rank~$2$ with $4$ elements, denoted by $U_4^2$, which is the unique forbidden minor for binary matroids~\cite{Tutte58a}.
In this light, the following corollary is not surprising.

\begin{corollary}[Section~4.5 in~\cite{Truemper98}]
  A matroid is either series-parallel or contains the graphic matroid $W_3$ or the uniform matroid $U_4^2$ as a minor.
\end{corollary}

\section{Computational study}
\label{sec_computations}

We start by describing the changes that we made for our implementation in the C programming language.
First, we decouple hashing of rows from hashing of columns by maintaining two hash tables.
Second, instead of a random vector $p \in \N^n$ for hashing of columns we use a deterministic one, defined via
$p_i \coloneqq 3^i \text{ mod } 2^{62}$ for $i = 1,2,\dotsc,n$.
This makes the algorithm deterministic for the price of not running in linear time due to potential hash table collisions.
Note that our integer data types would in principle admit the range $[-2^{63},2^{63}) \cap \Z$.
However, we use one bit less to detect overflows, which allows us to ensure that vectors $x$ and $-x$ receive the same hash value (see \cref{sec_ternary}).
The implemented hash table uses separate chaining via singly linked lists in order to resolve collisions.
Hashing of rows is done in an identical way.
The implementation is part of the Combinatorial Matrix Recognition Library (CMR)~\cite{CMR}.
Moreover, the paper is supplemented by the scripts and the specific version of the code with which the presented results were obtained.
All the experiments were carried out on a \SI{3.0}{\giga\hertz} Intel Xeon Gold 5217 CPU on a system with \SI{64}{\giga\byte} of RAM memory.
During all our experiments, our hash table implementation never encountered any collisions.

\subsection{Matrices from mixed-integer optimization}
\label{sec_computations_mip}

\DeclareDocumentCommand\integer{m}{\num[round-precision=0]{#1}}

\begin{table}[htpb]
  \caption{%
    Results for 33 of 79 ternary coefficient submatrices of mixed-integer-programs from the MIPLIB~\cite{MIPLIB2017} that are series-parallel.
    Depicted are characteristics of the ternary submatrix of the constraint matrix as well as the running times of the total unimodularity test with and without \cref{algo_reduce}.
  }
  \label{table_mip_sp}
  \renewcommand{\arraystretch}{1.05}
  \setlength{\tabcolsep}{3pt}
  \sisetup{round-mode=places, round-precision=2}%
  \begin{center}
    \begin{tabular}{lrrrrr}
      \hline
      \textbf{Instance} & \multicolumn{3}{c}{\textbf{Ternary submatrix}} & \multicolumn{1}{c}{\textbf{Without}} & \multicolumn{1}{c}{\textbf{With}} \\
      & \textbf{rows} & \textbf{cols} & \textbf{nonzeros} & \multicolumn{1}{c}{\textbf{\cref{algo_reduce}}} & \multicolumn{1}{c}{\textbf{\cref{algo_reduce}}} \\
      \hline
\texttt{blp-ar98} & \integer{913} & \integer{16021} & \integer{15806} & \SI{0.017859}{\s} & \SI{0.0183}{\s} \\
\texttt{blp-ic98} & \integer{627} & \integer{13640} & \integer{13550} & \SI{0.015228}{\s} & \SI{0.015919}{\s} \\
\texttt{bnatt500} & \integer{6525} & \integer{4024} & \integer{10053} & \SI{0.094027}{\s} & \SI{0.005173}{\s} \\
\texttt{bppc4-08} & \integer{20} & \integer{1456} & \integer{1454} & \SI{0.002925}{\s} & \SI{0.002122}{\s} \\
\texttt{csched007} & \integer{103} & \integer{1758} & \integer{1562} & \SI{0.0061}{\s} & \SI{0.006078}{\s} \\
\texttt{csched008} & \integer{111} & \integer{1536} & \integer{1405} & \SI{0.007964}{\s} & \SI{0.003343}{\s} \\
\texttt{exp-1-500-5-5} & \integer{300} & \integer{990} & \integer{1480} & \SI{0.031854}{\s} & \SI{0.002597}{\s} \\
\texttt{fhnw-binpack4-48} & \integer{910} & \integer{3675} & \integer{3640} & \SI{0.011413}{\s} & \SI{0.013939}{\s} \\
\texttt{fhnw-binpack4-4} & \integer{152} & \integer{507} & \integer{494} & \SI{0.004194}{\s} & \SI{0.004176}{\s} \\
\texttt{glass4} & \integer{142} & \integer{256} & \integer{263} & \SI{0.00242}{\s} & \SI{0.002432}{\s} \\
\texttt{lotsize} & \integer{725} & \integer{2985} & \integer{4175} & \SI{0.105616}{\s} & \SI{0.002428}{\s} \\
\texttt{mas74} & \integer{1} & \integer{151} & \integer{150} & \SI{0.000557}{\s} & \SI{0.000567}{\s} \\
\texttt{mas76} & \integer{1} & \integer{151} & \integer{150} & \SI{0.000438}{\s} & \SI{0.000575}{\s} \\
\texttt{mik-250-20-75-4} & \integer{120} & \integer{270} & \integer{120} & \SI{0.004167}{\s} & \SI{0.004142}{\s} \\
\texttt{milo-v12-6-r2-40-1} & \integer{4914} & \integer{1764} & \integer{6296} & \SI{0.073547}{\s} & \SI{0.006823}{\s} \\
\texttt{n3div36} & \integer{4426} & \integer{22120} & \integer{27960} & \SI{0.876898}{\s} & \SI{0.044153}{\s} \\
\texttt{neos-2657525-crna} & \integer{254} & \integer{524} & \integer{676} & \SI{0.013024}{\s} & \SI{0.001567}{\s} \\
\texttt{neos-2978193-inde} & \integer{332} & \integer{20800} & \integer{20800} & \SI{0.016134}{\s} & \SI{0.012915}{\s} \\
\texttt{neos-3004026-krka} & \integer{4225} & \integer{17030} & \integer{16900} & \SI{0.045863}{\s} & \SI{0.046488}{\s} \\
\texttt{neos-3381206-awhea} & \integer{4} & \integer{2375} & \integer{1900} & \SI{0.004659}{\s} & \SI{0.00508}{\s} \\
\texttt{neos-3754480-nidda} & \integer{300} & \integer{201} & \integer{300} & \SI{0.004189}{\s} & \SI{0.004525}{\s} \\
\texttt{neos-4338804-snowy} & \integer{441} & \integer{1344} & \integer{2562} & \SI{0.006949}{\s} & \SI{0.007501}{\s} \\
\texttt{neos-4413714-turia} & \integer{952} & \integer{190401} & \integer{190004} & \SI{0.097634}{\s} & \SI{0.098469}{\s} \\
\texttt{neos-4647030-tutaki} & \integer{2800} & \integer{12600} & \integer{11200} & \SI{0.036046}{\s} & \SI{0.032161}{\s} \\
\texttt{neos-4954672-berkel} & \integer{1764} & \integer{567} & \integer{1035} & \SI{0.02295}{\s} & \SI{0.008276}{\s} \\
\texttt{neos-787933} & \integer{133} & \integer{236376} & \integer{61944} & \SI{0.312689}{\s} & \SI{0.312252}{\s} \\
\texttt{neos-848589} & \integer{737} & \integer{550539} & \integer{550539} & \SI{0.265529}{\s} & \SI{0.267024}{\s} \\
\texttt{neos17} & \integer{1} & \integer{535} & \integer{50} & \SI{0.002478}{\s} & \SI{0.002442}{\s} \\
\texttt{pk1} & \integer{30} & \integer{86} & \integer{60} & \SI{0.001886}{\s} & \SI{0.000305}{\s} \\
\texttt{proteindesign121hz512p9} & \integer{79} & \integer{159145} & \integer{159133} & \SI{0.078653}{\s} & \SI{0.078704}{\s} \\
\texttt{proteindesign122trx11p8} & \integer{67} & \integer{127326} & \integer{127315} & \SI{0.064747}{\s} & \SI{0.065054}{\s} \\
\texttt{supportcase42} & \integer{18439} & \integer{18442} & \integer{34831} & \SI{63.50751}{\s} & \SI{0.025572}{\s} \\
\texttt{tr12-30} & \integer{360} & \integer{1080} & \integer{1068} & \SI{0.025282}{\s} & \SI{0.002093}{\s} \\ \hline
    \end{tabular}
  \end{center}
\end{table}

\begin{table}[htpb]
  \caption{%
    Results for 46 of 79 ternary coefficient submatrices of mixed-integer-programs from the MIPLIB~\cite{MIPLIB2017} that are not series-parallel.
    Depicted are characteristics of the ternary submatrix of the constraint matrix and its SP-reduced submatrix as well as the running times of the total unimodularity test with and without \cref{algo_reduce}.
  }
  \label{table_mip_nosp}
  \renewcommand{\arraystretch}{1.05}
  \setlength{\tabcolsep}{3pt}
  \sisetup{round-mode=places, round-precision=1}%
  \begin{center}
    \begin{tabular}{lrrrrrrrr}
    \hline
      \textbf{Instance} & \multicolumn{3}{c}{\textbf{Ternary submatrix}} & \multicolumn{3}{c}{\textbf{SP-reduced submatrix}} & \multicolumn{1}{c}{\textbf{Without}} & \multicolumn{1}{c}{\textbf{With}} \\
      & \textbf{rows} & \textbf{cols} & \textbf{nonzeros} & \textbf{rows} & \textbf{cols} & \textbf{nonzeros} & \multicolumn{1}{c}{\textbf{\cref{algo_reduce}}} & \multicolumn{1}{c}{\textbf{\cref{algo_reduce}}} \\
      \hline
\texttt{30n20b8} & \integer{90} & \integer{18380} & \integer{18438} & \integer{27} & \integer{22} & \integer{54} & \SI{0.011019}{\s} & \SI{0.010993}{\s} \\
\texttt{50v-10} & \integer{50} & \integer{2013} & \integer{732} & \integer{50} & \integer{183} & \integer{366} & \SI{0.021119}{\s} & \SI{0.01035}{\s} \\
\texttt{assign1-5-8} & \integer{31} & \integer{156} & \integer{260} & \integer{31} & \integer{130} & \integer{260} & \SI{0.003748}{\s} & \SI{0.00372}{\s} \\
\texttt{b1c1s1} & \integer{2848} & \integer{3872} & \integer{8112} & \integer{2080} & \integer{2192} & \integer{5728} & \SI{1.577017}{\s} & \SI{1.466902}{\s} \\
\texttt{beasleyC3} & \integer{500} & \integer{2500} & \integer{2500} & \integer{301} & \integer{426} & \integer{852} & \SI{0.195069}{\s} & \SI{0.015168}{\s} \\
\texttt{binkar10\_1} & \integer{1016} & \integer{2298} & \integer{4326} & \integer{785} & \integer{1242} & \integer{3030} & \SI{0.015872}{\s} & \SI{0.012717}{\s} \\
\texttt{cbs-cta} & \integer{10112} & \integer{22326} & \integer{54520} & \integer{244} & \integer{11163} & \integer{22326} & \SI{96.590535}{\s} & \SI{8.555163}{\s} \\
\texttt{cost266-UUE} & \integer{1389} & \integer{4161} & \integer{8151} & \integer{1332} & \integer{2052} & \integer{4104} & \SI{0.070218}{\s} & \SI{0.012924}{\s} \\
\texttt{drayage-100-23} & \integer{405} & \integer{11025} & \integer{22050} & \integer{210} & \integer{10960} & \integer{21920} & \SI{7.631541}{\s} & \SI{7.650101}{\s} \\
\texttt{drayage-25-23} & \integer{405} & \integer{11025} & \integer{22050} & \integer{210} & \integer{10960} & \integer{21920} & \SI{7.88959}{\s} & \SI{7.643947}{\s} \\
\texttt{fiball} & \integer{3449} & \integer{34218} & \integer{34648} & \integer{221} & \integer{258} & \integer{688} & \SI{0.055529}{\s} & \SI{0.054523}{\s} \\
\texttt{h80x6320d} & \integer{238} & \integer{12640} & \integer{18881} & \integer{80} & \integer{3160} & \integer{6320} & \SI{3.339353}{\s} & \SI{0.555658}{\s} \\
\texttt{ic97\_potential} & \integer{1046} & \integer{205} & \integer{2092} & \integer{523} & \integer{205} & \integer{1046} & \SI{0.090329}{\s} & \SI{0.005124}{\s} \\
\texttt{icir97\_tension} & \integer{368} & \integer{2494} & \integer{2027} & \integer{243} & \integer{145} & \integer{1078} & \SI{0.050254}{\s} & \SI{0.010911}{\s} \\
\texttt{lectsched-5-obj} & \integer{33526} & \integer{5651} & \integer{59340} & \integer{7880} & \integer{127} & \integer{15760} & \SI{81.153469}{\s} & \SI{1.49771}{\s} \\
\texttt{mad} & \integer{30} & \integer{220} & \integer{400} & \integer{30} & \integer{200} & \integer{400} & \SI{0.007793}{\s} & \SI{0.008008}{\s} \\
\texttt{mc11} & \integer{400} & \integer{3040} & \integer{3040} & \integer{400} & \integer{760} & \integer{1520} & \SI{0.184517}{\s} & \SI{0.015445}{\s} \\
\texttt{neos-1122047} & \integer{57791} & \integer{5000} & \integer{115580} & \integer{57790} & \integer{5000} & \integer{115580} & \SI{37.014054}{\s} & \SI{37.16188}{\s} \\
\texttt{neos-1445765} & \integer{1024} & \integer{20617} & \integer{20663} & \integer{10} & \integer{12} & \integer{24} & \SI{0.165528}{\s} & \SI{0.019002}{\s} \\
\texttt{neos-1582420} & \integer{10080} & \integer{10100} & \integer{22407} & \integer{2387} & \integer{2407} & \integer{7021} & \SI{59.486599}{\s} & \SI{33.566895}{\s} \\
\texttt{neos-3024952-loue} & \integer{3390} & \integer{3255} & \integer{7725} & \integer{3315} & \integer{3075} & \integer{7575} & \SI{0.067966}{\s} & \SI{0.05687}{\s} \\
\texttt{neos-3046615-murg} & \integer{258} & \integer{274} & \integer{546} & \integer{18} & \integer{32} & \integer{64} & \SI{0.002618}{\s} & \SI{0.004071}{\s} \\
\texttt{neos-3083819-nubu} & \integer{4719} & \integer{8644} & \integer{20118} & \integer{3648} & \integer{6628} & \integer{17031} & \SI{25.060069}{\s} & \SI{19.333827}{\s} \\
\texttt{neos-3627168-kasai} & \integer{1190} & \integer{1448} & \integer{3038} & \integer{1022} & \integer{924} & \integer{2674} & \SI{0.128488}{\s} & \SI{0.101996}{\s} \\
\texttt{neos-4387871-tavua} & \integer{554} & \integer{4004} & \integer{7984} & \integer{279} & \integer{1012} & \integer{2024} & \SI{0.393238}{\s} & \SI{0.017358}{\s} \\
\texttt{neos-4738912-atrato} & \integer{918} & \integer{6189} & \integer{13035} & \integer{768} & \integer{2025} & \integer{5955} & \SI{350.488336}{\s} & \SI{125.508562}{\s} \\
\texttt{neos-5107597-kakapo} & \integer{6498} & \integer{138} & \integer{13008} & \integer{1698} & \integer{114} & \integer{3396} & \SI{3.396775}{\s} & \SI{0.112903}{\s} \\
\texttt{neos-911970} & \integer{59} & \integer{888} & \integer{1680} & \integer{59} & \integer{840} & \integer{1680} & \SI{0.05078}{\s} & \SI{0.051475}{\s} \\
\texttt{nexp-150-20-8-5} & \integer{2385} & \integer{20115} & \integer{22350} & \integer{150} & \integer{2119} & \integer{4238} & \SI{0.245972}{\s} & \SI{0.170972}{\s} \\
\texttt{opm2-z10-s4} & \integer{160625} & \integer{6250} & \integer{321250} & \integer{160625} & \integer{6250} & \integer{321250} & \SI{862.460508}{\s} & \SI{861.90701}{\s} \\
\texttt{p200x1188c} & \integer{200} & \integer{2376} & \integer{2376} & \integer{200} & \integer{594} & \integer{1188} & \SI{0.116153}{\s} & \SI{0.011227}{\s} \\
\texttt{radiationm18-12-05} & \integer{33265} & \integer{40623} & \integer{73885} & \integer{14905} & \integer{18372} & \integer{36948} & \SI{695.08632}{\s} & \SI{27.964045}{\s} \\
\texttt{rd-rplusc-21} & \integer{125896} & \integer{118} & \integer{124457} & \integer{4} & \integer{4} & \integer{8} & \SI{1.635799}{\s} & \SI{0.096006}{\s} \\
\texttt{reblock115} & \integer{4715} & \integer{1150} & \integer{9430} & \integer{4706} & \integer{1140} & \integer{9412} & \SI{0.355385}{\s} & \SI{0.355089}{\s} \\
\texttt{rmatr100-p10} & \integer{7260} & \integer{7359} & \integer{21877} & \integer{7260} & \integer{7259} & \integer{21777} & \SI{17.001843}{\s} & \SI{13.344478}{\s} \\
\texttt{rmatr200-p5} & \integer{37617} & \integer{37816} & \integer{113048} & \integer{37617} & \integer{37616} & \integer{112848} & \SI{867.375674}{\s} & \SI{791.631219}{\s} \\
\texttt{rococoB10-011000} & \integer{1442} & \integer{4456} & \integer{7787} & \integer{450} & \integer{2025} & \integer{4050} & \SI{0.053426}{\s} & \SI{0.019739}{\s} \\
\texttt{rococoC10-001000} & \integer{1088} & \integer{3117} & \integer{5396} & \integer{330} & \integer{1353} & \integer{2706} & \SI{0.040037}{\s} & \SI{0.011736}{\s} \\
\texttt{sct2} & \integer{2001} & \integer{5884} & \integer{8901} & \integer{1797} & \integer{2885} & \integer{5770} & \SI{6.775848}{\s} & \SI{43.227286}{\s} \\
\texttt{sp150x300d} & \integer{150} & \integer{600} & \integer{600} & \integer{150} & \integer{253} & \integer{506} & \SI{0.014844}{\s} & \SI{0.009828}{\s} \\
\texttt{supportcase18} & \integer{120} & \integer{13410} & \integer{14400} & \integer{30} & \integer{30} & \integer{70} & \SI{1.013893}{\s} & \SI{0.001932}{\s} \\
\texttt{supportcase26} & \integer{870} & \integer{40} & \integer{1700} & \integer{434} & \integer{40} & \integer{868} & \SI{0.068498}{\s} & \SI{0.011473}{\s} \\
\texttt{swath1} & \integer{503} & \integer{6805} & \integer{19121} & \integer{102} & \integer{6240} & \integer{18640} & \SI{0.558728}{\s} & \SI{0.284716}{\s} \\
\texttt{swath3} & \integer{503} & \integer{6805} & \integer{19121} & \integer{102} & \integer{6240} & \integer{18640} & \SI{0.560856}{\s} & \SI{0.285831}{\s} \\
\texttt{trento1} & \integer{1259} & \integer{7687} & \integer{73662} & \integer{1189} & \integer{6415} & \integer{69590} & \SI{13.745141}{\s} & \SI{5.657498}{\s} \\
\texttt{uccase9} & \integer{49061} & \integer{21446} & \integer{47270} & \integer{8008} & \integer{8008} & \integer{19327} & \SI{0.904189}{\s} & \SI{0.80043}{\s} \\ \hline
    \end{tabular}
  \end{center}
\end{table}

The most important application of total unimodularity lies in mixed-integer optimization since coefficient matrices of integer programs that are totally unimodular have very attractive properties~\cite{HoffmanK56}.
Moreover, many substructures of coefficient matrices arising in practice are totally unimodular, e.g., if they encode network flows or (conditional) precedences.
For this reason we created a set of ternary matrices for which testing for total unimodularity is of potential interest for mixed-integer programming.
We consider the currently 240 benchmark MIP instances from the state-of-the-art benchmark library MIPLIB~2017~\cite{MIPLIB2017}.
For each of these instances, we extracted a large ternary submatrix in a greedy fashion: successively remove a row or a column with a maximum number entries not in $\{-1,0,+1\}$.
For the matrices of the nine instances \texttt{enlight\_hard}, \texttt{gen-ip002}, \texttt{gen-ip054}, \texttt{markshare2}, \texttt{markshare\_4\_0}, \texttt{ns1952667}, \texttt{pg}, \texttt{pg5\_34} and \texttt{timtab1} our greedy removal algorithm returned a matrix without any nonzeros.
Moreover, for the four instances \texttt{neos-4763324-toguru}, \texttt{radiationm40}, \texttt{square41} and \texttt{square47} the total unimodularity test took more than an hour.
Finally, for another four instances \texttt{buildingenergy}, \texttt{neos-873061}, \texttt{snp-02-004-104} and \texttt{thor50dday} disabling the preprocessing with \cref{algo_reduce} forced the implementation to carry out extremely many $2$-sum decompositions (which were actually SP-reductions).
Since the resulting matrices of each $2$-sum decomposition are created in memory, the code ran into memory problems despite the \SI{64}{\giga\byte} available RAM.
Note that this problem does not occur when using \cref{algo_reduce} since only the SP-reduced submatrix is explicitly created in memory for further computations.
For a further 144 instance, the computed ternary submatrices did not pass the very first step of the total unimodularity test, which is a test whether the sign pattern of the nonzeros of the given ternary matrix is \emph{balanced}.
For details we refer to the respective descriptions in~\cite{Truemper90,WalterT13}.

For the remaining 79 matrices we ran the total unimodularity test with and without preprocessing using \cref{algo_reduce}.
In \cref{table_mip_sp,table_mip_nosp} we present the sizes of the ternary submatrices, the SP-reduced submatrices and report about the timings of the total unimodularity test in both runs.
Note that for the series-parallel submatrices we do not report about the SP-reduced matrices since these are empty.

For series-parallel submatrices, both running times are very short since the total unimodularity test does not need to do anything further besides searching for $2$-separations and carrying out $2$-sum decompositions.
In particular, running times are very similar.
Only for larger matrices the preprocessing actually gives a speed-up.
The reason is that an explicit application of a $2$-sum decomposition is more costly for larger matrices.

Among the more interesting instances with non-series-parallel ternary submatrices we can find more with longer running times.
Here, after carrying out 2-sum decompositions, further algorithms are applied, among them the enumeration of 3-separations which is responsible for the asymptotic running time of the total unimodularity test.
Here, in several cases, preprocessing reduces the running time significantly.
Moreover, it rarely slows down the computation, except for one remarkable instance: for \texttt{sct2}, detailed inspection reveals that the version without preprocessing spends about \SI{2.5}{\s} for finding 2-separations 3750 times, while the SP-reduction algorithm is only called 348 times, which requires \SI{40.1}{\s}.
Profiling indicated that after carrying out the SP-reductions, subsequent decisions are made differently than in the run without preprocessing.

These instances show that submatrices arising in practical mixed-integer programs admit SP-reductions, and are sometimes even series-parallel.
Moreover, we can conclude that in most cases the application of \cref{algo_reduce} is advantageous.

\subsection{Random matrices}
\label{sec_computations_random}

\begin{figure}[htpb]
  \centering
  \subfloat[][%
    Effect of reduction type on running time.
  ]{%
    \begin{tikzpicture}
      \begin{axis}[
        ybar stacked,
        enlargelimits=0.05,
        ymajorgrids,
        xmin=0,
        xmax=1.0,
        ymin=0,
        ymax=18,
        xtick={0.0,0.1,0.2,0.3,0.4,0.5,0.6,0.7,0.8,0.9,1.0},
        xticklabels={$0$,$\frac{1}{10}$,$\frac{2}{10}$,$\frac{3}{10}$,$\frac{4}{10}$,$\frac{5}{10}$,$\frac{6}{10}$,$\frac{7}{10}$,$\frac{8}{10}$,$\frac{9}{10}$,$1$},
        ytick={0,2,4,6,8,10,12,14,16,18},
        x label style={yshift=-3mm},
        width=0.46\textwidth,
        height=65mm,
        bar width=1.4mm,
        ylabel={Running time [s]},
        xlabel={Fraction $\gamma = 1-\beta$ of copy extensions},
        legend entries={Binary reduction,Ternary reduction,Graphicness},
        legend cell align=left,
        legend pos=north west,
        legend style={font=\small},
        ]
          \addplot+[blue!80!black, fill=blue, bar shift=-1.7mm] coordinates {
            (0.0, 0.005)
            (0.1, 0.005)
            (0.2, 0.006)
            (0.3, 0.006)
            (0.4, 0.008)
            (0.5, 0.014)
            (0.6, 0.035)
            (0.7, 0.100)
            (0.8, 0.358)
            (0.9, 1.326)
            (1.0, 5.674)
          };

          \resetstackedplots

          \addplot+[red!80!black, fill=red, bar shift=0mm, postaction={pattern=north east lines}] coordinates {
            (0.0, 0.005)
            (0.1, 0.005)
            (0.2, 0.006)
            (0.3, 0.007)
            (0.4, 0.009)
            (0.5, 0.014)
            (0.6, 0.033)
            (0.7, 0.103)
            (0.8, 0.370)
            (0.9, 1.574)
            (1.0, 6.123)
          };

          \resetstackedplots

          \addplot+[cyan!80!black, fill=cyan, bar shift=1.7mm, postaction={pattern=north west lines}] coordinates {
             (0.0,0.032)
             (0.1,0.034)
             (0.2,0.035)
             (0.3,0.038)
             (0.4,0.042)
             (0.5,0.054)
             (0.6,0.084)
             (0.7,0.220)
             (0.8,0.836)
             (0.9,3.720)
             (1.0,16.749)
          };

      \end{axis}
    \end{tikzpicture}
  }
  \hfill
  \subfloat[][%
    Effect of reduction type on running time per nonzero.\label{fig_comp_series_parallel_per_nonzero}
  ]{%
    \begin{tikzpicture}
      \begin{axis}[
        enlargelimits=0.05,
        grid=major,
        xmin=0,
        xmax=1.0,
        xtick={0.0,0.1,0.2,0.3,0.4,0.5,0.6,0.7,0.8,0.9,1.0},
        xticklabels={$0$,$\frac{1}{10}$,$\frac{2}{10}$,$\frac{3}{10}$,$\frac{4}{10}$,$\frac{5}{10}$,$\frac{6}{10}$,$\frac{7}{10}$,$\frac{8}{10}$,$\frac{9}{10}$,$1$},
        ytick={0,200,400,600,800,1000,1200,1400,1600},
        x label style={yshift=-3mm},
        ymin=0,
        ymax=1600,
        width=0.46\textwidth,
        height=65mm,
        ylabel={Running time per nonzero [ns]},
        xlabel={Fraction $\gamma = 1-\beta$ of copy extensions},
        legend entries={Binary,Ternary,Graphicness},
        legend cell align=left,
        legend pos=north east,
        legend style={font=\small},
        ]
        \addplot[blue,mark size=1mm,only marks] coordinates {
          (0.0,254.415)
          (0.1,243.886)
          (0.2,215.672)
          (0.3,185.732)
          (0.4,162.386)
          (0.5,138.960)
          (0.6,119.685)
          (0.7,98.614)
          (0.8,80.228)
          (0.9,68.103)
          (1.0,58.033)
        };

        \addplot[mark=diamond*, mark size=1mm, red, only marks] coordinates {
          (0.0,262.330)
          (0.1,252.633)
          (0.2,226.265)
          (0.3,196.478)
          (0.4,166.679)
          (0.5,137.745)
          (0.6,121.823)
          (0.7,102.793)
          (0.8,85.754)
          (0.9,73.340)
          (1.0,62.587)
        };

        \addplot[mark=triangle*,cyan,mark size=1mm,only marks] coordinates {
          (0.0,1582.953)
          (0.1,1489.558)
          (0.2,1335.169)
          (0.3,1100.803)
          (0.4,804.722)
          (0.5,530.022)
          (0.6,295.272)
          (0.7,205.278)
          (0.8,186.253)
          (0.9,179.435)
          (1.0,167.491)
        };
      \end{axis}
    \end{tikzpicture}
  }
  \caption{%
    Results for random series-parallel matrices with $\alpha = \delta = 0$, $p=1$, $\beta = 1 - \gamma$ and varying values for $\gamma \in [0,1]$.
    Averaged over 1000 matrices per configuration.
  }
  \label{fig_comp_series_parallel}
\end{figure}

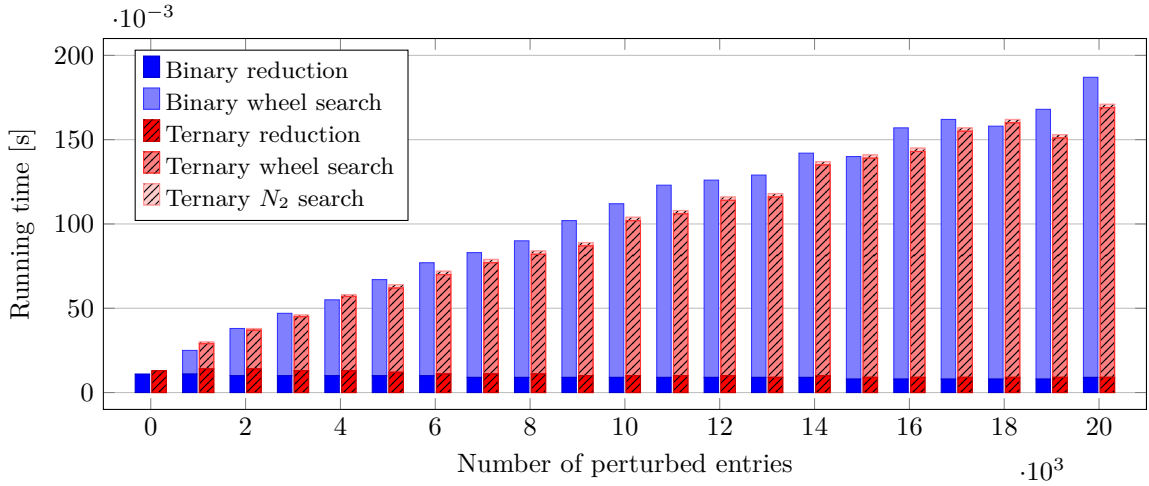
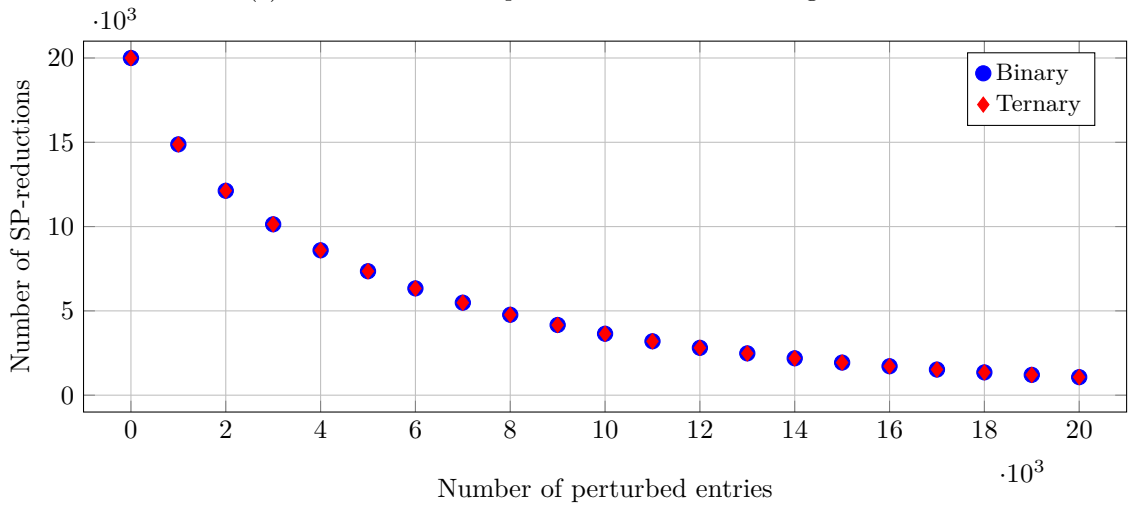
\begin{figure}[htpb]
  \centering
  \subfloat[][%
    Effect of number of perturbed entries on running time.\label{fig_comp_perturbed_time}
  ]{%
    \begin{tikzpicture}
      \begin{axis}[
        ybar stacked,
        enlargelimits=0.05,
        ymajorgrids,
        xmin=0,
        xmax=20000,
        ymin=0,
        ymax=0.2,
        ytick={0,0.05,0.1,0.15,0.2},
        width=0.90\textwidth,
        height=65mm,
        scaled x ticks={base 10:-3},
        scaled y ticks={real:0.001},
        ytick scale label code/.code={$\cdot 10^{-3}$},
        bar width=1.9mm,
        ylabel={Running time [s]},
        xlabel={Number of perturbed entries},
        legend entries={Binary reduction,Binary wheel search,Ternary reduction,Ternary wheel search,Ternary $N_2$ search},
        legend cell align=left,
        legend pos=north west,
        legend style={font=\small},
        ]
          \addplot+[blue!80!black, fill=blue, bar shift=-1.1mm] coordinates {
            (0,0.011)
            (1000,0.011)
            (2000,0.010)
            (3000,0.010)
            (4000,0.010)
            (5000,0.010)
            (6000,0.010)
            (7000,0.009)
            (8000,0.009)
            (9000,0.009)
            (10000,0.009)
            (11000,0.009)
            (12000,0.009)
            (13000,0.009)
            (14000,0.009)
            (15000,0.008)
            (16000,0.008)
            (17000,0.008)
            (18000,0.008)
            (19000,0.008)
            (20000,0.009)
          };

          \addplot+[blue!80!white, fill=blue!50!white, bar shift=-1.1mm] coordinates {
            (0,0.000)
            (1000,0.014)
            (2000,0.028)
            (3000,0.037)
            (4000,0.045)
            (5000,0.057)
            (6000,0.067)
            (7000,0.074)
            (8000,0.081)
            (9000,0.093)
            (10000,0.103)
            (11000,0.114)
            (12000,0.117)
            (13000,0.120)
            (14000,0.133)
            (15000,0.132)
            (16000,0.149)
            (17000,0.154)
            (18000,0.150)
            (19000,0.160)
            (20000,0.178)
          };

          \resetstackedplots

          \addplot+[red!80!black, fill=red, bar shift=1.1mm, postaction={pattern=north east lines}] coordinates {
            (0,0.013)
            (1000,0.014)
            (2000,0.014)
            (3000,0.013)
            (4000,0.013)
            (5000,0.012)
            (6000,0.011)
            (7000,0.011)
            (8000,0.011)
            (9000,0.010)
            (10000,0.010)
            (11000,0.010)
            (12000,0.010)
            (13000,0.009)
            (14000,0.010)
            (15000,0.009)
            (16000,0.009)
            (17000,0.009)
            (18000,0.009)
            (19000,0.009)
            (20000,0.009)
          };

          \addplot+[red!80!white, fill=red!50!white, bar shift=1.1mm, postaction={pattern=north east lines}] coordinates {
            (0,0.000)
            (1000,0.015)
            (2000,0.023)
            (3000,0.032)
            (4000,0.044)
            (5000,0.050)
            (6000,0.059)
            (7000,0.066)
            (8000,0.071)
            (9000,0.077)
            (10000,0.092)
            (11000,0.096)
            (12000,0.104)
            (13000,0.107)
            (14000,0.125)
            (15000,0.130)
            (16000,0.134)
            (17000,0.146)
            (18000,0.151)
            (19000,0.142)
            (20000,0.160)
          };

          \addplot+[red!40!white, fill=red!20!white, bar shift=1.1mm, postaction={pattern=north east lines}] coordinates {
            (0,0.000)
            (1000,0.001)
            (2000,0.001)
            (3000,0.001)
            (4000,0.001)
            (5000,0.002)
            (6000,0.002)
            (7000,0.002)
            (8000,0.002)
            (9000,0.002)
            (10000,0.002)
            (11000,0.002)
            (12000,0.002)
            (13000,0.002)
            (14000,0.002)
            (15000,0.002)
            (16000,0.002)
            (17000,0.002)
            (18000,0.002)
            (19000,0.002)
            (20000,0.002)
          };

      \end{axis}
    \end{tikzpicture}
  }

  \subfloat[][%
    Effect of number of perturbed entries on the number of possible SP-reductions.\label{fig_comp_perturbed_reductions}
  ]{%
    \begin{tikzpicture}
      \begin{axis}[
        enlargelimits=0.05,
        grid=major,
        xmin=0,
        scaled y ticks={base 10:-3},
        scaled x ticks={base 10:-3},
        x label style={yshift=-3mm},
        ymin=0,
        width=0.90\textwidth,
        height=65mm,
        ylabel={Number of SP-reductions},
        xlabel={Number of perturbed entries},
        legend entries={Binary,Ternary},
        legend cell align=left,
        legend pos=north east,
        legend style={font=\small},
        ]
        \addplot[blue,mark size=1mm,only marks] coordinates {
          (0,20000.0)
          (1000,14877.741)
          (2000,12124.935)
          (3000,10135.856)
          (4000,8593.215)
          (5000,7352.385)
          (6000,6334.444)
          (7000,5486.651)
          (8000,4770.08)
          (9000,4164.125)
          (10000,3644.12)
          (11000,3197.414)
          (12000,2808.424)
          (13000,2480.052)
          (14000,2188.203)
          (15000,1933.983)
          (16000,1717.079)
          (17000,1519.384)
          (18000,1351.854)
          (19000,1204.716)
          (20000,1068.677)
        };
        \addplot[mark=diamond*,mark size=1mm,red,only marks] coordinates {
          (0,20000.0)
          (1000,14874.641)
          (2000,12126.381)
          (3000,10138.297)
          (4000,8593.632)
          (5000,7349.927)
          (6000,6333.315)
          (7000,5484.406)
          (8000,4769.443)
          (9000,4160.348)
          (10000,3640.581)
          (11000,3195.554)
          (12000,2811.165)
          (13000,2477.014)
          (14000,2187.057)
          (15000,1934.97)
          (16000,1715.851)
          (17000,1521.537)
          (18000,1351.857)
          (19000,1203.337)
          (20000,1072.733)
        };
      \end{axis}
    \end{tikzpicture}
  }
  \caption{%
    Results for random series-parallel matrices with $\alpha = 0$, $\beta = \gamma = 0.5$, $p=1$ and varying values for $\delta \in [0,2]$.
    Averaged over 1000 matrices per configuration.}
  \label{fig_comp_perturbed}
\end{figure}

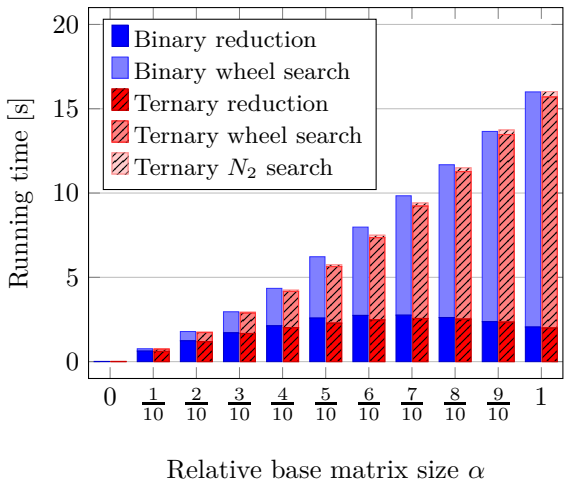
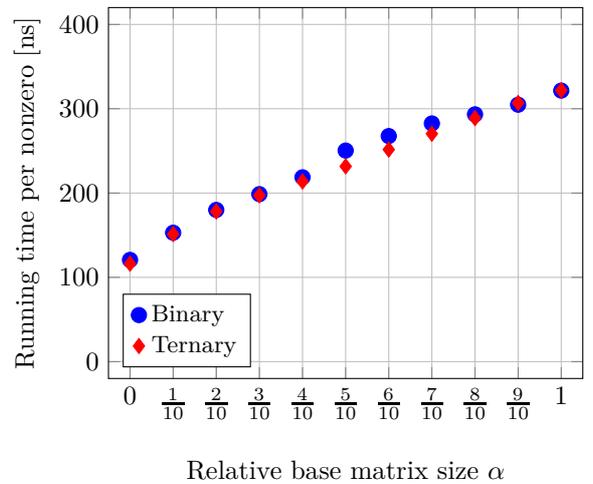
\begin{figure}[htpb]
  \centering
  \subfloat[][%
    Effect of base matrix size on running time.\label{fig_comp_varybase_time}
  ]{%
    \begin{tikzpicture}
      \begin{axis}[
        ybar stacked,
        enlargelimits=0.05,
        ymajorgrids,
        xmin=0,
        xmax=1.0,
        ymin=0,
        ymax=20,
        xtick={0.0,0.1,0.2,0.3,0.4,0.5,0.6,0.7,0.8,0.9,1.0},
        xticklabels={$0$,$\frac{1}{10}$,$\frac{2}{10}$,$\frac{3}{10}$,$\frac{4}{10}$,$\frac{5}{10}$,$\frac{6}{10}$,$\frac{7}{10}$,$\frac{8}{10}$,$\frac{9}{10}$,$1$},
        ytick={0,5,10,15,20},
        x label style={yshift=-3mm},
        width=0.46\textwidth,
        height=65mm,
        bar width=2mm,
        ylabel={Running time [s]},
        xlabel={Relative base matrix size $\alpha$},
        legend entries={Binary reduction,Binary wheel search,Ternary reduction,Ternary wheel search,Ternary $N_2$ search},
        legend cell align=left,
        legend pos=north west,
        legend style={font=\small},
        ]
          \addplot+[blue!80!black, fill=blue, bar shift=-1.1mm] coordinates {

            (0.0,0.012)
            (0.1,0.629)
            (0.2,1.241)
            (0.3,1.715)
            (0.4,2.131)
            (0.5,2.587)
            (0.6,2.741)
            (0.7,2.758)
            (0.8,2.616)
            (0.9,2.373)
            (1.0,2.057)
          };
          \addplot+[blue!80!white, fill=blue!50!white, bar shift=-1.1mm] coordinates {
            (0.0,0.000)
            (0.1,0.135)
            (0.2,0.545)
            (0.3,1.242)
            (0.4,2.214)
            (0.5,3.634)
            (0.6,5.237)
            (0.7,7.079)
            (0.8,9.062)
            (0.9,11.280)
            (1.0,13.941)
          };

          \resetstackedplots

          \addplot+[red!80!black, fill=red, bar shift=1.1mm, postaction={pattern=north east lines}] coordinates {
            (0.0,0.011)
            (0.1,0.606)
            (0.2,1.184)
            (0.3,1.669)
            (0.4,2.015)
            (0.5,2.282)
            (0.6,2.475)
            (0.7,2.559)
            (0.8,2.522)
            (0.9,2.343)
            (1.0,1.994)
          };
          \addplot+[red!80!white, fill=red!50!white, bar shift=1.1mm, postaction={pattern=north east lines}] coordinates {
            (0.0,0.000)
            (0.1,0.132)
            (0.2,0.540)
            (0.3,1.214)
            (0.4,2.147)
            (0.5,3.357)
            (0.6,4.881)
            (0.7,6.669)
            (0.8,8.750)
            (0.9,11.128)
            (1.0,13.705)
          };
          \addplot+[red!40!white, fill=red!20!white, bar shift=1.1mm, postaction={pattern=north east lines}] coordinates {
            (0.0,0.000)
            (0.1,0.013)
            (0.2,0.034)
            (0.3,0.055)
            (0.4,0.079)
            (0.5,0.110)
            (0.6,0.144)
            (0.7,0.185)
            (0.8,0.227)
            (0.9,0.274)
            (1.0,0.322)
          };
      \end{axis}
    \end{tikzpicture}
  }
  \hfill
  \subfloat[][%
    Effect of base matrix size on running time per nonzero.\label{fig_comp_varybase_per_nonzero}
  ]{%
    \begin{tikzpicture}
      \begin{axis}[
        enlargelimits=0.05,
        grid=major,
        xmin=0,
        xmax=1.0,
        xtick={0.0,0.1,0.2,0.3,0.4,0.5,0.6,0.7,0.8,0.9,1.0},
        xticklabels={$0$,$\frac{1}{10}$,$\frac{2}{10}$,$\frac{3}{10}$,$\frac{4}{10}$,$\frac{5}{10}$,$\frac{6}{10}$,$\frac{7}{10}$,$\frac{8}{10}$,$\frac{9}{10}$,$1$},
        x label style={yshift=-3mm},
        ymin=0,
        ymax=400,
        width=0.46\textwidth,
        height=65mm,
        ylabel={Running time per nonzero [ns]},
        xlabel={Relative base matrix size $\alpha$},
        legend entries={Binary,Ternary},
        legend cell align=left,
        legend pos=south west,
        legend style={font=\small},
        ]
        \addplot[blue,mark size=1mm,only marks] coordinates {
          (0.0,120.693)
          (0.1,152.928)
          (0.2,179.966)
          (0.3,198.689)
          (0.4,218.672)
          (0.5,250.384)
          (0.6,267.574)
          (0.7,282.561)
          (0.8,293.352)
          (0.9,304.800)
          (1.0,321.510)
        };
        \addplot[mark=diamond*,mark size=1mm,red,only marks] coordinates {
          (0.0,116.204)
          (0.1,151.270)
          (0.2,178.158)
          (0.3,197.307)
          (0.4,213.706)
          (0.5,231.639)
          (0.6,251.571)
          (0.7,270.334)
          (0.8,288.838)
          (0.9,306.824)
          (1.0,321.819)
        };
      \end{axis}
    \end{tikzpicture}
  }
  \caption{%
    Results for random series-parallel matrices with $\delta = 0$, $p=0.5$, $\beta = \gamma = \tfrac{1}{2}(1-\alpha)$ and varying values for $\alpha \in [0,1]$.
    Averaged over 100 matrices per configuration.}
  \label{fig_comp_varybase}
\end{figure}

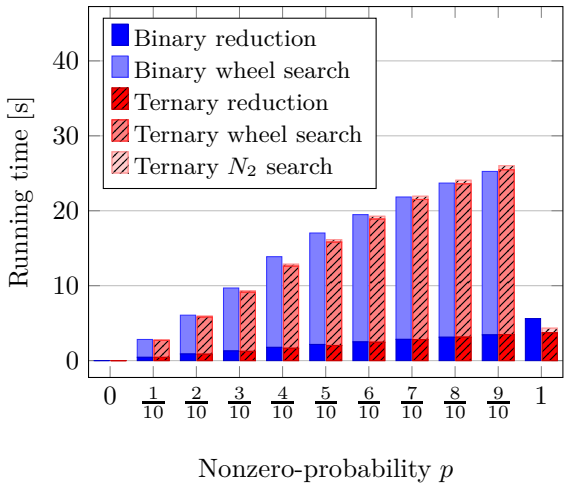
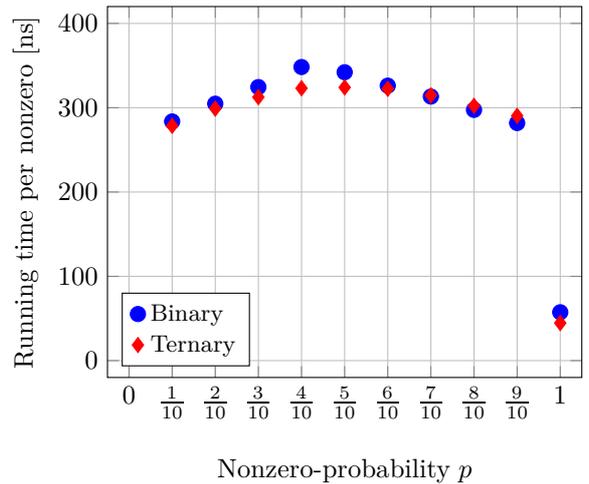
\begin{figure}[htpb]
  \centering
  \subfloat[][%
    Effect of nonzero probability on running time.
    \label{fig_comp_bigbase_time}
  ]{%
    \begin{tikzpicture}
      \begin{axis}[
        ybar stacked,
        enlargelimits=0.05,
        ymajorgrids,
        xmin=0,
        xmax=1.0,
        ymin=0,
        ymax=45,
        xtick={0.0,0.1,0.2,0.3,0.4,0.5,0.6,0.7,0.8,0.9,1.0},
        xticklabels={$0$,$\frac{1}{10}$,$\frac{2}{10}$,$\frac{3}{10}$,$\frac{4}{10}$,$\frac{5}{10}$,$\frac{6}{10}$,$\frac{7}{10}$,$\frac{8}{10}$,$\frac{9}{10}$,$1$},
        ytick={0,10,20,30,40},
        x label style={yshift=-3mm},
        width=0.46\textwidth,
        height=65mm,
        bar width=2mm,
        ylabel={Running time [s]},
        xlabel={Nonzero-probability $p$},
        legend entries={Binary reduction,Binary wheel search,Ternary reduction,Ternary wheel search,Ternary $N_2$ search},
        legend cell align=left,
        legend pos=north west,
        legend style={font=\small},
        ]
          \addplot+[blue!80!black, fill=blue, bar shift=-1.1mm] coordinates {
            (0.0,0.004)
            (0.1,0.442)
            (0.2,0.895)
            (0.3,1.302)
            (0.4,1.770)
            (0.5,2.155)
            (0.6,2.510)
            (0.7,2.837)
            (0.8,3.124)
            (0.9,3.446)
            (1.0,5.611)
          };
          \addplot+[blue!80!white, fill=blue!50!white, bar shift=-1.1mm] coordinates {
            (0.0,0.000)
            (0.1,2.380)
            (0.2,5.167)
            (0.3,8.387)
            (0.4,12.099)
            (0.5,14.873)
            (0.6,16.973)
            (0.7,18.998)
            (0.8,20.567)
            (0.9,21.802)
            (1.0,0.000)
          };

          \resetstackedplots

          \addplot+[red!80!black, fill=red, bar shift=1.1mm, postaction={pattern=north east lines}] coordinates {
            (0.0,0.004)
            (0.1,0.426)
            (0.2,0.847)
            (0.3,1.230)
            (0.4,1.640)
            (0.5,2.017)
            (0.6,2.442)
            (0.7,2.808)
            (0.8,3.124)
            (0.9,3.481)
            (1.0,3.717)
          };
          \addplot+[red!80!white, fill=red!50!white, bar shift=1.1mm, postaction={pattern=north east lines}] coordinates {
            (0.0,0.000)
            (0.1,2.268)
            (0.2,4.965)
            (0.3,7.904)
            (0.4,10.967)
            (0.5,13.798)
            (0.6,16.445)
            (0.7,18.697)
            (0.8,20.464)
            (0.9,21.987)
            (1.0,0.000)
          };
          \addplot+[red!40!white, fill=red!20!white, bar shift=1.1mm, postaction={pattern=north east lines}] coordinates {
            (0.0,0.000)
            (0.1,0.073)
            (0.2,0.135)
            (0.3,0.199)
            (0.4,0.260)
            (0.5,0.321)
            (0.6,0.377)
            (0.7,0.431)
            (0.8,0.486)
            (0.9,0.554)
            (1.0,0.611)
          };
      \end{axis}
    \end{tikzpicture}
  }
  \hfill
  \subfloat[][%
    Effect of nonzero probability on running time per nonzero.
    \label{fig_comp_bigbase_per_nonzero }
  ]{%
    \begin{tikzpicture}
      \begin{axis}[
        enlargelimits=0.05,
        grid=major,
        xmin=0,
        xmax=1.0,
        xtick={0.0,0.1,0.2,0.3,0.4,0.5,0.6,0.7,0.8,0.9,1.0},
        xticklabels={$0$,$\frac{1}{10}$,$\frac{2}{10}$,$\frac{3}{10}$,$\frac{4}{10}$,$\frac{5}{10}$,$\frac{6}{10}$,$\frac{7}{10}$,$\frac{8}{10}$,$\frac{9}{10}$,$1$},
        x label style={yshift=-3mm},
        ymin=0,
        ymax=400,
        width=0.46\textwidth,
        height=65mm,
        ylabel={Running time per nonzero [ns]},
        xlabel={Nonzero-probability $p$},
        legend entries={Binary,Ternary},
        legend cell align=left,
        legend pos=south west,
        legend style={font=\small},
        ]
        \addplot[blue,mark size=1mm,only marks] coordinates {
          (0.1,283.704)
          (0.2,304.751)
          (0.3,324.492)
          (0.4,348.269)
          (0.5,342.087)
          (0.6,326.161)
          (0.7,313.236)
          (0.8,297.433)
          (0.9,281.835)
          (1.0,57.379)
        };
        \addplot[mark=diamond*,mark size=1mm,red,only marks] coordinates {
          (0.1,278.525)
          (0.2,298.944)
          (0.3,312.562)
          (0.4,323.103)
          (0.5,324.101)
          (0.6,322.429)
          (0.7,314.702)
          (0.8,302.250)
          (0.9,290.374)
          (1.0,44.539)
        };
      \end{axis}
    \end{tikzpicture}
  }
  \caption{%
    Results for random series-parallel matrices with $\alpha = 1$, $\beta = \gamma = \delta = 0$ and varying values for $p \in [0,1]$.
    Averaged over 10 matrices per configuration.}
  \label{fig_comp_bigbase}
\end{figure}

In order to investigate the factors that have impact on the running time of our algorithms we generated random $N$-by-$N$ matrices in a structured way.
The matrices are parameterized by four factors $\alpha$, $\beta$, $\gamma$ and $\delta$ and a probability $p$.
The first factor $\alpha \in [0,1]$ indicates the portion of an $(\alpha \cdot N)$-by-$(\alpha \cdot N)$ \emph{base matrix} whose entries were drawn from a uniform distribution.
For binary matrices, $p$ indicates the probability of a $1$, while for ternary matrices, $+1$ and $-1$ each occur with probability $p/2$.
With high probability, this base matrix does not admit SP-reductions, unless $p$ is close to $0$ or $1$ or $\alpha$ is close to $0$.

This base matrix is extended by subsequently adding $(\beta \cdot N)$ unit rows, $(\beta \cdot N)$ unit columns, $(\gamma \cdot N)$ copies of rows and $(\gamma \cdot N)$ copies of columns to it.
For unit vectors the place of the unique $1$-entry is chosen uniformly at random.
Similarly, for copied vectors the source vector of the copy is chosen uniformly at random among the available ones.
To obtain the right matrix size always have $\alpha + \beta + \gamma = 1$.
The final step in our generation procedure is perturbation as proposed by one the anonymous reviewers:
in order to destroy a small amount of series-parallel structure we flip $(\delta \cdot N)$-many entries.
To stay within the (memory) capacity of our computational environment we carried out all our experiments for $N = \SI{10000}{}$.

\paragraph{Series-parallel matrices.}
The first set of random matrices has $\alpha=0$ (technically, there exists a $1$-by-$1$ base matrix), varying $\beta$, $\gamma = 1-\beta$ and $\delta = 0$.
Hence, all these matrices are series-parallel, and hence $2N$ reductions are carried out each time.
The results are shown in \cref{fig_comp_series_parallel}, also in comparison to our implementation of the graphicness algorithm due to Bixby and Wagner~\cite{BixbyW88}.
It is easy to see that our specialized algorithm performs significantly better than the graphicness test.
However, we would like to emphasize that the main motivation for \cref{algo_reduce} is not to test for being series-parallel faster than the graphicness test, but also to compute the (largest) SP-reduced submatrix in case of a non-series-parallel matrix, as done for the matrices from mixed-integer programming.
In particular, the graphicness test does not produce such a submatrix.
Hence, we omit the comparison for subsequent matrices since these are mostly non-series-parallel.

Now let us turn to the actual analysis of our algorithm's behavior.
The visible running time increase in case of a larger number of copy reductions seems to be surprisingly high.
Hence, we also depict the average running time per nonzero.
As \cref{fig_comp_series_parallel_per_nonzero} shows, this value even decreases with a larger number of copy reductions.
This can be explained by the fact that the algorithm spends some time per SP-reduction and some time per nonzero that is processed for such a reduction.
Obviously, a unit reduction requires to process only one nonzero.

\paragraph{Perturbed series-parallel matrices.}
For the second set we kept $\alpha = 0$, $\beta = \gamma = 0.5$ and $p=1$, but now varied $\delta$, i.e., we considered random perturbations of series-parallel matrices.
Results are shown in \cref{fig_comp_perturbed}.
Clearly, the number of possible SP-reductions decreases with increasing $\delta$, as \cref{fig_comp_perturbed_reductions} shows.
This is natural since the first perturbed entries most likely destroys the corresponding row and the corresponding column reduction at once.
The running times of \cref{algo_reduce} are not affected drastically, in particular since the overall number of nonzeros of the matrices do not change much.
However, it is apparent that the running time of \cref{algo_wheel_search} is dominated by the size of the SP-reduced submatrix.

\paragraph{Matrices with varying base matrix sizes.}
For the third set we kept $\delta = 0$, $p = 0.5$, varied $\alpha$ and set $\beta = \gamma = \tfrac{1}{2}(1-\alpha)$, i.e., we considered different sizes of base matrices that were extended to size $N$-by-$N$.
For all the generated matrices, exactly $(1-\alpha) \cdot 2N$ many SP-reductions reductions could be applied, i.e., the all generated base matrices were SP-reduced.
The results are shown in \cref{fig_comp_varybase}.
While the effort for \cref{algo_wheel_search} increases with $\alpha$, the running time for \cref{algo_reduce} starts to decrease again for larger $\alpha$, which can be justified by the fact that fewer reductions must actually applied.
As in the previous experiments, the overall running time is dominated by that of \cref{algo_wheel_search}.
In particular, searching an $N_2$-submatrix takes almost no time.
Finally, while our implementation was slower for ternary series-parallel matrices than for binary ones, the situation is reversed for non-series-parallel matrices.
We do not have an explanation for this, but the differences are also not very significant.

\paragraph{Matrices without SP-reductions.}
Our last set of randomly generated matrices has $\alpha = 1$, i.e., an $N$-by-$N$ base matrix, $\beta = \gamma = \delta = 0$, and varying nonzero probabilities $p$.
For $p=0$ and for $p=1$ (in the binary case) the resulting matrices were the all-zeros and all-ones matrices, respectively, and thus series-parallel.
In all other cases, all of the generated matrices were SP-reduced.
Moreover, all returned forbidden submatrices were $M_3$ or $M_3'$, except for the ternary case with $p=1$, where $N_2$ was returned.
Further results are shown in \cref{fig_comp_bigbase}.
Both, the effort for \cref{algo_reduce} for determining that the matrix is SP-reduced, and the effort for \cref{algo_wheel_search} for finding a wheel submatrix grow with the increasing nonzero probability, which underlines our previous explanation that the number of nonzeros is the driving factor.
As before, the overall running time is dominated by the search for wheel submatrices.
The jump for $p=1$ is due to the special structure:
for binary matrices, the algorithm has to apply SP-reductions, which is generally fast, while for ternary matrices it has to determine that no ternary SP-reduction is necessary and inspect a single binary SP-reduction in order to find an $N_2$-submatrix.

\section{Conclusion}
\label{sec_conclusion}

Our presented algorithm for the problem of recognizing series-parallel matrices is not only theoretically fast.
We have shown experimentally that is is also very effective in practice and can deal with large matrices.
In fact the measured running times were in the same order of magnitude as elementary operations such as transposing the matrix (assuming that it is given in a sparse data structure).
This means that with this direct test one can recognize series-parallel matrices much faster than by determining the graph first and testing the latter for being series-parallel.

Moreover, our algorithm and is certifying, and actually computing a certificate for not being series-parallel requires linear time as well.
In practice, this step actually needs most of the running time.
As a byproduct, our algorithm is able to determine a maximal non-SP submatrix.
This allows to use it as a preprocessing for more complex algorithms such as the recognition of totally unimodular matrices.
In fact we can conclude that for several matrices from mixed-integer programming the bottleneck of the total unimodularity test was indeed the application of SP-reductions.

\paragraph{Acknowledgements.}
The author is grateful to the anonymous reviewers whose comments led to improvements of this manuscript, in particular for the suggestions of further matrix classes.
This publication is part of the project \emph{Making Mixed-Integer Programming Solvers Smarter and Faster using Network Matrices} (with project number \emph{OCENW.M20.151} of the research programme \emph{NWO Open Competition Domain Science -- M} which is (partly) financed by the Dutch Research Council (NWO).

\bibliographystyle{plainurl}
\bibliography{series-parallel-matroid-recognition}

\end{document}